\documentclass[journal,10pt,twocolumn]{IEEEtran}

\usepackage{amssymb, amsmath, amsthm}
\usepackage{amsfonts, bbm}
\usepackage{graphicx}
\usepackage{epstopdf}
\usepackage{times}
\usepackage{float} 
\usepackage{array}
\usepackage{arydshln}
\usepackage{bibentry}
\usepackage{booktabs}

\usepackage{cite}
\usepackage{subfigure}
\usepackage{enumitem}
\usepackage{color}
\usepackage{setspace}
\usepackage{bigstrut}
\usepackage{multirow}
\usepackage{algorithm}
\usepackage{algorithmic}
\usepackage{subfig}
\usepackage{mathtools}

\newcommand*{\Scale}[2][4]{\scalebox{#1}{$#2$}}%
\newtheorem{Corollary}{Corollary}

\newtheorem{theorem}{Theorem}
\newcommand{\figref}[1]{Fig.~\ref{#1}}
\begin{document}

\title{On the Performance of Renewable Energy-Powered UAV-Assisted Wireless Communications}

\author{Silvia Sekander, Hina Tabassum,  and Ekram Hossain\thanks{S. Sekander and E. Hossain are with the Department of Electrical and Computer Engineering, University of Manitoba, Canada (Emails: sekandes@myumanitoba.ca, Ekram.Hossain@umanitoba.ca). H. Tabassum is with the Lassonde School of Engineering,  York University, Canada (Email: Hina.Tabassum@lassonde.yorku.ca). This work was supported by the Natural Sciences and Engineering Research Council of Canada (NSERC).}
}

\maketitle
\begin{abstract}
	We develop novel statistical models of the  harvested energy from renewable energy sources (such as solar and wind energy)  considering  harvest-store-consume (HSC)  architecture.  We consider three renewable energy harvesting scenarios, i.e. (i)~harvesting from the solar power, (ii)~harvesting from the wind power, and (iii)~hybrid solar and wind power. In this context, we first derive the closed-form expressions for the probability density function (PDF) and cumulative density function (CDF) of the harvested power from the solar and wind energy sources. Based on the derived expressions, we calculate the probability of energy outage at UAVs and signal-to-noise ratio (SNR) outage at ground cellular users. The energy outage occurs when the UAV is unable to support the flight consumption and transmission consumption from its battery power and the harvested power. Due to the intricate distribution of the hybrid solar and wind power, we derive novel closed-form expressions for the  moment generating function (MGF) of the harvested  solar power and wind power. Then, we apply Gil-Pelaez inversion  to evaluate the energy outage at the UAV and signal-to-noise-ratio (SNR) outage at the ground users. We formulate the SNR outage minimization problem and obtain closed-form solutions for the transmit power and flight time of the UAV.  In addition, we derive novel closed-form expressions for the moments of the solar power and wind power and demonstrate their applications in computing novel performance metrics considering  the stochastic nature of the amount of harvested energy as well as energy arrival time. These performance metrics include  the probability of charging the UAV battery within the flight time, average UAV battery charging time, probability of energy outage at UAVs,  and the probability of eventual energy outage (i.e. the probability of energy outage in a finite duration of time) at UAVs. Numerical results validate the analytical expressions and reveal interesting insights related to the optimal flight time and transmit power of the UAV as a function of the harvested energy. 

\end{abstract}


\section{Introduction}
The unmanned aerial vehicles (UAVs) will play an important role in fulfilling the communication requirements of the next generation wireless networks~\cite{DSC}. The primary benefits of UAV networks include the operation in dangerous and disastrous environments, on-demand relocation, improved coverage due to higher line-of-sight (LOS) connections with the ground cellular users, and an extra-degree of freedom due to three-dimensional (3D) movements~\cite{WCU}. Despite various potential benefits, the energy consumption at UAVs is a primary bottleneck due to their mechanical and communication power requirements. Energy consumption at UAVs can potentially limit the endurance time and communication performance of UAVs. To enhance the sustainability and endurance time of UAVs,  energy harvesting from renewable energy sources (e.g. solar, wind, electromagnetic radiations) is a low-cost alternative where a device can harvest  from free energy sources. However, renewable energy sources are intermittent and uncertain and thus the amount of energy harvested is  random~\cite{24}. For deployment and operation of the renewable energy-powered UAV-assisted wireless communications systems, it will be useful to mathematically model the dynamics of energy harvested from a variety of renewable energy sources and analyze the communication performance as a function of the parameters such as time of the day, wind speed, etc.

\subsection{Background Work}
Several studies in the existing literature have dealt with the energy efficiency of UAV based wireless communications. The authors in \cite{OTT} investigated the power minimization problem in a UAV network such that the cell boundaries and locations of UAVs can be optimized iteratively. Given a certain cell boundary, the locations of the UAVs were derived using a facility location framework. In \cite{MIT}, the authors discussed an energy-efficient UAV deployment method to collect information from internet-of-things (IoT) devices in uplink. Using the $K$-means clustering approach, the ground devices are first clustered and then served by the UAV.  The authors in \cite{3DU} proposed an optimal placement algorithm for UAV BSs such that the coverage to a ground BS and the network energy efficiency can be maximized. In \cite{emc}, authors have proposed a coverage model considering multi-UAV system to achieve energy-efficient communication. They have solved the problem in two steps: coverage
maximization and power control and proved that both the problems fall in the category of exact potential games (EPG). Finally, they have devised an algorithm to do 
energy-efficient coverage deployment using spatial adaptive play (MUECD-SAP) method.

Nonetheless, the aforementioned research works have not explored the benefits of energy harvesting for UAV-assisted wireless communications. A preliminary study in \cite{END} discussed the  concept of Energy Neutral Internet of Drones (enIoD)  to achieve enhanced connectivity by overcoming energy limitations for longer endurance and continuous operation. The authors considered wireless power transfer  to energize the UAVs and thus reducing the gap in harvested and consumed energy. They have also conceptualized special UAVs which are able to carry energy from one charging station to the other using the concept of opportunistic charging (OC). The authors in \cite{extract} proposed the radio energy harvesting at the UAV to improve the endurance time of the UAV. Dirty paper coding and information-theoretic uplink-downlink channel duality  was considered to maximize the network throughput. In \cite{outage},  UAV-based relaying  was considered with energy harvesting capability at the UAV. The outage probability was derived considering different urban environment parameters. Here, the harvested energy comes from the ground base station (GBS). In \cite{resource}, the authors  investigated the resource allocation problem for UAV-assisted networks, where the UAV provides radio frequency energy to the  device-to-device (D2D) pairs. A resource allocation problem to maximize the average throughput of the UAV-assisted D2D network was formulated as a non-convex optimization problem considering the energy causality constraints. In \cite{wpt},  a UAV-enabled two-user wireless power transfer system was considered where the UAV charges multiple energy receivers for specific time period. Via optimization of the UAV trajectory subject to its maximum speed constraints, the authors have minimized the transferred energy to the intended receivers. Authors in \cite{uav-enabled} have solved the similar problem considering multi-user UAV system.

The aforementioned research works are focused on considering energy harvesting through radio frequency sources instead of renewable energy sources. The primary benefits of renewable energy harvesting over the distance-dependent wireless powered networks include (i)~the reduced consumption of network resources (e.g., transmission channel, transmission time, and transmit power) and (ii)~power transfer  between the energy transmitter and the energy receiver is independent of the distance between them. 

Recently, 
in \cite{optimal3D}, the authors  considered maximizing the sum throughput over a given period of time for a solar-powered UAV systems. A mixed-integer non-convex optimization problem was formulated considering energy harvesting, aerodynamic power consumption, finite energy storage system, and quality-of-service (QoS) requirements of the users. They have used monotonic optimization to solve the non-convex problem and attain the optimal 3D-trajectory along with resource allocation. In \cite{spatial}, an energy management framework was proposed for cellular heterogeneous networks (HetNets) supported by solar powered drones to jointly find the optimal trips of the UAVs and the ground BSs that can be turned off to minimize the total energy consumption of the network. UAVs are able to charge their batteries either at a charging station or from harvested solar energy.  Another research work \cite{solar} formulated a framework for energy management in UAV-assisted HetNets. The UAVs have the provision for solar energy harvesting as well as energy charging from fixed charging stations. They have studied the optimal deployment of UAVs to minimize energy consumption in the network.

\subsection{Motivation and Contributions}
The aforementioned research works do not incorporate accurate  models to characterize the renewable energy harvested from solar and wind sources (as a function of specific solar and wind parameters) and therefore rely on  assumptions (e.g. solar energy is modeled as a Gamma random variable in \cite{solar}) for tractability reasons. A plethora of research works consider solar energy harvesting in wireless sensor networks~\cite{harvesting,combining,solarharvested}; however, to the best of our knowledge, there are no concrete statistical models for the solar or wind harvested energy or power and their applications to communication networks are unknown.  In addition, the aforementioned  research works  are focused mainly on the numerical optimization or simulation-based studies. Subsequently, the impact of the network parameters such as time of the day, velocity of the wind, solar radiance cannot be captured on the energy outage at UAVs and transmission outage at ground users. 

The contributions of this paper are summarized as follows:
\begin{itemize}
	\item We develop novel statistical models for the  amount of harvested energy  considering three renewable energy harvesting scenarios, i.e.  (i) solar power, (ii) wind power, and (iii) hybrid solar and wind power. Based on the derived models, we calculate the probability of energy outage at the UAV and signal-to-noise ratio (SNR) outage at ground cellular users.
	\item We derive the closed-form expressions for the probability density function (PDF) and cumulative density function (CDF) of the harvested power from the solar and wind energy sources.  Due to the intricate distribution of the hybrid solar and wind power, we derive the closed-form expressions for the  moment generating function (MGF) of the harvested  solar power and wind power. Then, we applied Gil-Pelaez inversion  to evaluate the energy outage at the UAV and SNR outage at users. 
	\item We formulate the SNR outage minimization problem and obtain closed-form solutions for the transmit power and flight time of the UAV.
	\item We derive the closed-form expressions for the moments of the harvested wind power and solar power and demonstrate their applications in computing new performance metrics considering the scenario when both the amount of energy as well as the energy arrival time is  stochastic. That is, not only the amount of energy is random  but also the time of energy arrival is  random. These performance metrics include  the probability of harvesting energy within the flight duration, average battery charging time, probability of energy outage at the UAV,  and the probability of eventual energy outage (which is the probability of energy outage in a finite duration of time) at the UAV are analyzed. 
	\item Numerical results validate the analytical expressions by providing a comparison with the Monte-Carlo simulations and exhibit interesting insights related to the optimal flight time and transmit power of the UAV as a function of the harvested energy.  
\end{itemize}

\begin{table*}[!h]
\centering
\small
\caption{{\color{black} Mathematical notations} }
\resizebox{\textwidth}{!}{\begin{tabular}{|c|c|c|c|}
\hline
  \textbf{Notation}&\textbf{Description} & \textbf{Notation}&\textbf{Description}\\ \hline
  \footnotesize$T_f$;$T_b-T_f$&\footnotesize Flight time; Hover \& transmission duration & \footnotesize $P_b$;$P_d$&\footnotesize Back-up battery power; Downlink transmit power\\ \hline
  \footnotesize $I(t)$&\footnotesize Radiation intensity & $K_c$& \footnotesize Threshold radiation intensity\\ \hline
  \footnotesize$I_d(t)$&\footnotesize Deterministic fundamental intensity & \footnotesize$\Delta I(t)$& \footnotesize Stochastic attenuation\\ \hline
  \footnotesize$\eta_c$&\footnotesize PV system efficiency & $c;k$& \footnotesize Scale \& shape parameter of Weibull random variable\\ \hline
  \footnotesize$V_{\mathrm{ci}};V_r;V_{\mathrm{co}}$&\shortstack{\footnotesize Cut-in; rated; cut-off wind velocity} &  \footnotesize$A; \rho$& \footnotesize Rotor area; Air density\\ \hline
  $n_p$;$r_p$&\shortstack{\footnotesize Number of propellers; Propeller radius} & $P;P_w$ & \footnotesize Harvested solar power; Harvested wind power \\ \hline
  $E_c$&\shortstack{\footnotesize Energy consumed from battery} & $E_t$ & \footnotesize Energy required for transmission \\ \hline
  $E_f$&\shortstack{\footnotesize Energy required for flight} & $h_d$ & \footnotesize UAV altitude \\ \hline
   $\lambda$&\shortstack{\footnotesize Energy arrival rate} & $A_i$ & \footnotesize Inter-arrival times of the energy packets \\ \hline
  $U(t)$&\shortstack{\footnotesize Accumulated energy
in the battery} & $X_i$ & \footnotesize Energy packet size \\ \hline
$u_0$&\shortstack{\footnotesize Initial battery energy} & $\tau$ & \footnotesize Battery recharge time \\ \hline
$\phi(u_0)$&\shortstack{\footnotesize Eventual energy outage probability} & $r^*$ & \footnotesize Adjustment coefficient \\ \hline

\end{tabular}}
\label{Notation_Summary_mmwave}
\end{table*}
\vspace{2 mm}
\noindent\textbf{Notations}: 
$\Gamma(a)=\int_0^\infty x^{a-1} e^{-x} dx$ represents the Gamma function, ${\Gamma}_u (a;b)=\int_b^\infty x^{a-1} e^{-x} dx$ denotes the upper incomplete Gamma function, ${\Gamma}_l(a;b)=\int_0^b x^{a-1} e^{-x} dx$ denotes the lower incomplete Gamma function and ${\Gamma}(a;b_1;b_2)=\Gamma_u(a;b_1)-\Gamma_u(a;b_2)=\int_{b_1}^{b_2} x^{a-1} e^{-x} dx$ denotes the generalized Gamma function \cite{ebook}. $_2F_1[\cdot,\cdot,\cdot,\cdot]$ denotes the Gauss's hypergeometric function.
$\mathbb{P}(A)$ denotes the probability of event $A$. $f(\cdot)$, $F(\cdot)$, and $\mathcal{L}(\cdot)$ denote the probability density function (PDF), cumulative distribution function (CDF), and Laplace Transform, respectively. Finally, $\mathbb{U}(\cdot)$, $\delta(\cdot)$, and $\mathbb{E}[\cdot]$ denote the unit step function, the Dirac-delta function, and   the expectation operator, respectively. $\mathrm{erf}(x)$ is the error function  expressed as $
\mathrm{erf}(x) = \frac{2}{\sqrt{\pi}} \int\limits_0^x e^{-t^2}dt
$ and $\mathrm{erfc}(x) = 1-\mathrm{erf}(x)$ denotes the complementary error function \cite[8.25/4]{ebook}. A list of important variables is presented in Table~1.

The rest of the paper is organized as follows. The system model and assumptions are stated in Section II. The energy outage and the signal-to-noise ratio (SNR) outage probabilities are evaluated in Sections III and IV, respectively. Section V analyzes the moments of the harvested power and presents several applications of these moments including evaluation of  probability of battery charging, average battery charging time, and eventual energy outage probability. Section VI presents the numerical results before the paper is concluded in Section VII.  

\section{System Model and Assumptions}
In this section, we describe the network model, the air-to-ground (AtG) channel propagation model, the UAV energy consumption model, and the harvested energy models for solar and wind energy sources.

\subsection{Network Model}

We consider a UAV-enabled with solar and wind energy harvesting capability that serves ground cellular users on orthogonal transmission channels. The users are distributed uniformly  in a circular region of area $A =\pi R^2$, where $R$ represents the radius of the considered circular region. The UAV harvests energy from the solar and/or  wind energy  depending on the energy harvesting model. The UAV operates in two states:  (i) traveling to the desired location for transmission while harvesting energy, (ii) hovering and transmitting to the cellular users or  traveling back to the charging station (which is located at the origin)  to charge itself if needed.  A  duration of $T_b$ is considered in which the UAV travels for duration $T_f$  and hovers at the destination  for a duration $T_b-T_f$ to perform downlink transmission given there is no energy outage. Since the maximum distance a UAV can travel is $R_m$ in a straight line trajectory from the charging station, $T_b$ is set as $T_b=R_\mathrm{m}/v_d$. The UAV travels with the speed $v_d$ (in m/sec) and is equipped with the fixed back-up battery power $P_b$ to support the UAV flight (in case if the harvested power is not enough).  The UAV performs data transmission to the cellular users in the downlink using transmit power $P_d$ for the time duration $T_b - T_f$ given that there is no energy outage.

\begin{figure}[t]
	\begin{center}
		\includegraphics[scale=.5]{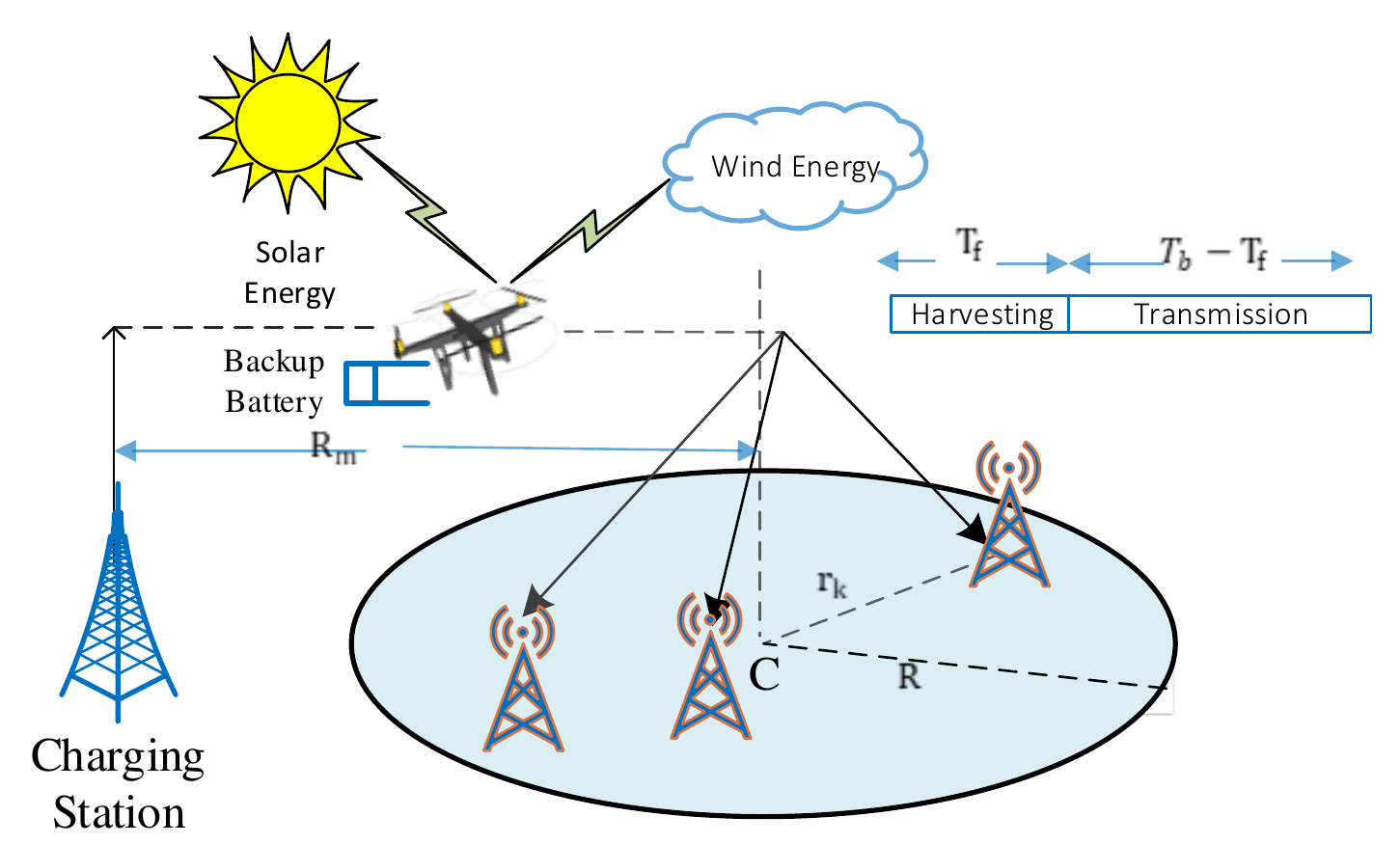}
	\end{center}
	\caption{System model. }
	\label{fig8}
\end{figure} 

\subsection{Air-to-Ground (AtG) Channel Model}
The RF signals generated by the UAV first travels through the free space until they reach the man-made urban environment, where additional losses (referred to as excessive path-loss) occur due to foliage and/or urban environment. The excessive path-loss is  random in nature and cannot be characterized by a well-known distribution. As such, the mean value of excessive path-loss $\eta_{\epsilon}$  obtained from empirical distribution fitting is typically considered. The RF transmissions from a given UAV fall into  three propagation groups, Line-of-Sight (LOS) propagation, non-LOS (NLOS) propagation via strong reflection and refraction, and a very limited contribution (less than 3\% as reported in \cite{MAG}) by the deep fading resulting from consecutive reflections and diffraction. As such, the third group has been discarded in most of the relevant research studies. Since the excessive path-loss depends largely on the first two propagation groups, $\eta_{\epsilon}$   can be considered as a constant that can be obtained by averaging all samples in a certain propagation group. The values of $\eta_{\epsilon}$ are listed for various frequencies and urban environments in \cite[Table II]{MAG}.

The AtG path-loss can thus be defined as follows~\cite{OLA}:
\begin{equation}
\mathrm{PL}_{\epsilon}=\mathrm{FSPL}+ \eta_{\epsilon},
\end{equation}
where $\epsilon \in \{\mathrm{LOS}, \mathrm{NLOS}\}$ and free space path-loss (FSPL) can be evaluated using the standard Friis equation, i.e.
$
\mathrm{FSPL}= 20 \mathrm{log}_{10} \left(\frac{4 \pi f_c d}{c}\right),
$
where $f_c$ is the carrier frequency (Hz), $c$ is the speed of light (m/s), and $d$ is the distance between the UAV and  the receiving user. The probability of having LOS for a user $i$ depends on the altitude $h_d$ of the serving UAV  and the horizontal distance between the UAV and $i^{\mathrm{th}}$ user, which is $r_i=\sqrt{(x_D-x_i)^2 + (y_D-y_i)^2}$. The $i^{\mathrm{th}}$ user is located at $(x_i,y_i,0)$ and the UAV is located at $(x_D,y_D,h_d)$. The LOS probability is thus given by: 
\begin{equation}
P_{\mathrm{LOS}}(h_d,r_i)=\frac{1}{{1+a \mathrm{exp}\left(-b\left(\mathrm{arctan}\left(\frac{h_d}{r_i}\right)-a\right)\right)}},
\end{equation}
where $\mathrm{arctan}(\frac{h_d}{r_i})$ is the elevation angle between the UAV and the served user (in degrees). Here $a$ and $b$ are constant values that depend on the choice of urban environment (high-rise urban, dense urban, sub-urban, urban). They are also known as S-curve parameters as they are obtained by approximating the LOS probability ({\em given by  International Telecommunication Union (ITU-R)} \cite{ITU,OLA}) with a simple modified
Sigmoid function (S-curve). Subsequently, the approximate LOS probability can be given for various urban environments while capturing the buildings' heights distribution, mean number of man made structures, and percentage of the built-up land area, of the considered urban environment.
The NLOS probability is given as
$
P_{\mathrm{NLOS}}(h_d,r_i)=1-P_{\mathrm{LOS}}(h_d,r_i).
$

The path-loss expression can then be written as~\cite{E3D}:
\begin{equation}
L(h_d,r_i)=20 \mathrm{log} \left(\sqrt{h_d^2+r_i^2}\right)+AP_{\mathrm{LOS}}(h_d,r_i)+B,
\label{eq}
\end{equation}
where $A = \eta_{\mathrm{LOS}}- \eta_{\mathrm{NLOS}}$, $B = 20 \mathrm{log}\left( \frac{4\pi f_c}{c} \right) + \eta_{\mathrm{NLOS}}$,
$\eta_{\mathrm{LOS}}$ and $\eta_{\mathrm{NLOS}}$ (in dB) are, respectively, the losses corresponding to the LOS and non-LOS reception depending on the environment.  The considered AtG propagation model can capture various environments (such as high-rise urban, dense urban, sub-urban, urban)~\cite{OLA, E3D, OTT, MIT, UAVU}. 

\subsection{Harvested Power Model}
\subsubsection{Solar Power Model}
The output  power from the Photo Voltaic (PV) system depends on the solar radiation intensity, the solar cell temperature, and the PV system efficiency. The output power of the PV system $P$ at time $t$ can therefore be modeled as follows\cite{RED}:
\[
    P= 
\begin{cases}
    \frac{\eta_c}{K_c}I(t)^2, &  0<I(t)<K_c\\
    \eta_c I(t),              & I(t) \geq K_c
\end{cases},
\label{powersolar}
\]
where $I(t)$ denotes the radiation intensity and $K_c$ is a threshold for the radiation intensity beyond which the efficiency $\eta_c$ can be approximated as a constant. Note that  $I(t)$ can be calculated as a sum of deterministic fundamental intensity $I_d(t)$ and stochastic attenuation  $\Delta I(t)$ due to weather effects as well as clouds occlusion, i.e. $I(t)=I_d(t)+\Delta I(t)$. Generally, $I_d(t)$ depends on the time of a day and the months/seasons of a year and can be found from the following equation\cite{RED}:
\[
    I_d(t)= 
\begin{cases}
    I_{\mathrm{max}}(-1/36 t^2+2/3 t-3),& \text{if } 6 \leq t <18\\
    0,              & \text{otherwise}
\end{cases}.
\]
The distribution of $\Delta I$ follows a standard normal distribution, i.e.
$\Delta I \sim N(0, 1)$. Then  the distribution of $I(t)$ is a normal distribution with shifted mean as expressed below:
\begin{equation}
f_I(I)=\frac{1}{\sqrt{2\pi}} \mathrm{exp}\left(\frac{(I-I_d)^2}{2}\right).
\end{equation}

\subsubsection{Wind Power Model}
To estimate the wind energy potential, the velocity of wind is a primary variable and is typically modeled by the Weibull distribution as shown below\cite{Reliability}:
\begin{equation}
f_V(v)=\frac{k}{c} \left(\frac{v}{c}\right)^{k-1} \mathrm{e}^{-(v/c)^k},
\end{equation}
where $v$ is the speed of wind, $c$ and $k$ denote the  scale and  shape parameter of the Weibull random variable, respectively. Physically speaking,  $c$ indicates the wind strength of the considered location and $k$ is the peak value of the wind distribution. The CDF of Weibull variable is given by
\begin{equation}
F_V(v)=1-\mathrm{exp}[-(v/c)^k],
\end{equation}
where  $k=(\frac{\sigma}{\mu})^{-1.086}$ and $c=\frac{\mu}{\Gamma (1+1/k)}$.

The energy produced by a wind turbine generator can be obtained by means of its power curve, where the relationship between the wind speed and the delivered power can be established as shown below\cite{Reliability}:
\[
    P_w = 
\begin{cases}
    q(v),   & V_{ci}<v\leq V_r\\
    P_r, & V_r<v \leq V_{co}\\
    0, & \mathrm{elsewhere}
\end{cases},
\]
where $P_w$ and $v$ denote the output power of the wind turbine and
the wind velocity, respectively. $V_{ci}, V_r$ and $V_{co}$ represent the cut-in wind velocity, rated wind velocity, and cut-off wind velocity, respectively. The non-linear part of the power curve can be defined as follows\cite{online}:
\begin{equation}
q(v)=\frac{1}{2} \rho A C v^3,
\end{equation}
where $C$ is a constant equivalent to the power coefficient, $A$ is the rotor area, and $\rho$ denotes the air density.

\subsection{UAV Energy Consumption Model}
Based on the momentum disk theory and blade element theory, the power consumption model for the UAV in hover state can be defined. This model considers the power consumption due to thrust $T$ which is defined as a force  to move an aircraft through the air. During hover, it can be assumed that the thrust  is approximately the same as the total weight force $W$ in Newton, i.e. $T=W$. In the case of a hovering aircraft, reaction force is approximately equal to the gravitational force. As a result, we can assume $W = m g$, where $m$ is the UAV mass (in kg) and $g$ denotes earth gravity (in m/s$^2$).
The mechanical power $P_{\mathrm{hov}}$ of a system that exerts a force $T$ on an object moving with velocity $v_d$ can be defined as~\cite{energy-efficient}:
\begin{equation}
P_{\mathrm{hov}}=T v_d = mg v_d =\sqrt{\frac{(m g)^3}{2 \pi r_p^2 n_p \rho}},
\end{equation}
where $A$ denotes the rotor disk area and $\rho$ is air density in kg$/$m$^3$, $A=\pi r_p^2 n_p$, $n_p$ and $r_p$ denote the number of propellers and propeller radius, respectively.

During transmission period $T_b-T_f$, the UAV has to spend the power for hovering as well as for transmitting to the desired users. On the other hand, during flight the UAV needs to hover and spend some energy to keep the UAV active. The energy consumption at UAVs can thus be modeled as follows:
\begin{itemize}
\item Energy consumed during transmission (while hovering), $E_t=(P_{\mathrm{hov}}+P_d)(T_b-T_f)$.
\item Energy consumed during return flight (no transmission), $E_f=(P_{\mathrm{hov}}+\gamma_d)T_f$.
\end{itemize}
where $\gamma_d$ denotes the activation energy, $T_f$ denotes the flight time, and $P_d$ denotes the transmit power of the UAV. 

\subsection{Actual Energy Consumption from the UAV Battery}
If the harvested solar energy $E_h=P T_f$ or harvested wind energy  $E_h=P_w T_f$ is high enough to support the $E_f$ and $E_t$, no energy needs to be consumed from the battery. Therefore, the net energy consumption (or energy cost) from the battery $E_c$ is zero. Another scenario can happen where the harvested energy levels are not sufficient. However, after combining with the UAV backup battery, the transmission and return flight  can still be supported. In this case, after spending all the harvested energy, we need to consume the deficit energy from the UAV battery. The net energy consumption becomes $E_c=E_t+E_f-E_h$, where $E_t$ is the energy required for transmission, $E_f$ is the energy required for the return flight, and $E_h$ denotes the harvested energy. Finally, if the harvested energy along with the battery energy cannot support transmission,  the UAV will not transmit and fly back to the charging station and consume only the energy for  return flight  $E_f$. Finally, there could be a scenario when the harvested energy and battery energy are insufficient to support the  return flight, however, we consider that the back-up battery is designed to support the emergency return flight. The net energy consumption is as follows:
\begin{equation}
E_c=
\begin{cases}
0, & E_h>E_t+E_f\\
E_t+E_f-E_h, & E_h+E_b>E_t+E_f\\
E_f-E_h,  & E_h+E_b>E_f
\end{cases}.
\end{equation}

\section{Characterization of the Energy Outage}
In this section, we characterize the energy outage probability of the UAV considering (i) solar power, (ii) wind power, and (iii) hybrid solar and wind power. We derive the PDF and CDF of the harvested energy from the solar and wind power sources and then  characterize their respective energy outages. For hybrid solar and wind power, we derive the Laplace Transforms of the solar and wind harvested power and use Gil-Pelaez inversion to characterize the energy outage.

The UAV is declared to be in energy outage if the battery energy along with the harvested energy is not enough to support the energy consumption of the return flight and the data transmission. Mathematically, the energy outage probability $E_{\mathrm{out}}$ can be defined as follows:
\begin{align}
E_{\mathrm{out}}&=\mathbb{P}[P_b T_b + H T_f < E_t+E_f],
\nonumber\\
&=\mathbb{P}\left[H<\frac{E_t+E_f  - P_b T_b}{T_f}\right],
\nonumber\\
&=\mathbb{P}\left[H<\theta\right],
\label{out1}
\end{align}
where $\theta=\frac{E_t+E_f  - P_b T_b}{T_f}$ and $H$ is the harvested power and $H=P$ if the source of energy is solar, $H=P_w$ if the source of energy is wind, and $H=P+P_w$ in case of harvesting from both solar and wind energy.

\subsection{Harvested Power - Solar and Wind}
The harvested power from solar energy is a function of the time of the day and months of the year, solar radiance, and efficiency of the photo-voltaic system.
The distribution of  the harvested solar power $P(t)$ can be written as in the following:
\begin{theorem}[Distribution of the Harvested Solar Power] The PDF of the harvested solar power can be given as follows:
\[
    f_P(p)= 
\begin{cases}
    \frac{1}{2} \sqrt{\frac{K_c}{\eta_c p}} f_I\left(\sqrt{\frac{K_c p}{\eta_c }}\right),&  p<\eta_c K_c\\
    \frac{1}{\eta_c}  f_I(p/\eta_c),              & \text{otherwise}.
\end{cases}
\]
\end{theorem}
\begin{proof}
See {\bf Appendix~A.}
\end{proof}
On the other hand, the distribution of the harvested power from wind energy can be derived as follows:
\begin{theorem}[Distribution of the Harvested Wind Power] Letting $a=\frac{1}{2}\rho A C$,
the PDF of the harvested wind power can be derived as follows:
\begin{align*}
&{f_{P_W}(p_w)=\frac{k {p_w}^{\frac{k}{3}-1} \mathrm{exp}\left(-\frac{p_w^{\frac{k}{3}}}{a^{\frac{k}{3}}c^k}\right)}{3 c^k a^{\frac{k}{3}}}  }\mathbb{U}(a V^3_{ci} < p_w < a V_r^3) +\\&\Scale[1]{\delta(p_w-P_r)(
F(V_{co})-F(V_r))+
\delta(p_w)(
F(V_{ci})+1-F(V_{co}))}.
\end{align*}
Essentially, the last two terms represent the case when the harvested power becomes equal to the maximum constant rated power and the case when the harvested power is zero.
\end{theorem}
\begin{proof}
See {\bf Appendix B.}
\end{proof}

\subsection{Energy Outage}
\subsubsection{Solar Power} In order to characterize the probability of  energy outage, first we derive the CDF $F_P(p)=\int_{0}^p f_P(t) dt$ of the harvested solar power using {\bf Theorem~1} as follows:
\begin{equation}
    F_P(p)= 
\begin{cases}
    \frac{1}{2} \left(\mathrm{erf}\left[\frac{I_d}{\sqrt{2}}\right] -\mathrm{erf}\left[\frac{I_d-\sqrt{\frac{K_c p}{\eta_c} }}{\sqrt{2}}\right]\right),& p<\eta_c K_c\\
    \frac{1}{2} \left(\mathrm{erf}\left[\frac{I_d-K_c}{\sqrt{2}}\right] +\mathrm{erf}\left[\frac{p-I_d \eta_c}{\sqrt{2} \eta_c}\right]\right),              & p>\eta_c K_c
\end{cases}.
\label{cdfsolar}
\end{equation}
Using the CDF of the harvested power in \eqref{cdfsolar} and the definition of energy outage event given in \eqref{out1}, we can evaluate the energy outage as follows:
\begin{equation}
E_{\mathrm{out}}= F_P\left(\theta\right).
\label{out2}
\end{equation}

\subsubsection{Wind Power} 
Similarly, the probability of energy outage with wind harvesting can be derived by  first  determining the CDF of the harvested wind power using {\bf Theorem~2} as:
\begin{align}
&F_{P_w}(p_w)=\left[\gamma \left(1,\frac{p_w^{\frac{k}{3}}}{a^{\frac{k}{3}}c^k} \right)- \gamma \left(1,\frac{V_{\mathrm{ci}}^{\frac{k}{3}}}{a^{\frac{k}{3}}c^k} \right)\right]+\nonumber\\
&\Scale[1]{\mathbb{U}(p_w)[F(V_{ci})+1-F(V_{co})]
+\mathbb{U}(p_w-P_r)[F(V_{co})-F(V_r)]}.
\label{out3}
\end{align}
Using the CDF of the harvested power in \eqref{out3} and the definition of energy outage event given in \eqref{out1}, we can evaluate the energy outage as $E_{\mathrm{out}}=F_{P_w}\left(\theta\right)$.

\subsubsection{Hybrid Solar-Wind Power}
Evaluating the PDF of the  sum of the harvested solar and wind power is not tractable due to the convolution of  the distributions of the solar and wind powers. Therefore, we utilize an MGF-based approach to evaluate the energy outage. That is, we propose to use Gil-Pelaez inversion theorem to characterize energy outage. Note that the solar and wind powers are independent random variables. Therefore, the Laplace Transform of the total harvested power  can be given by the product of the  Laplace Transforms of the solar power and wind power. For this, we first derive the Laplace Transforms (or MGF) of the solar power and wind power and then determine the energy outage using Gil-Pelaez inversion theorem.
\begin{theorem}[Laplace Transform of Solar Power $P$ and Wind Power $P_w$]
The Laplace transform $\mathcal{M}_P(s) =\int_0^\infty e^{-s P} f_P(p) dp$ of the  harvested power from solar energy can be derived as follows:
\begin{align}
&\mathcal{L}_P(s)=
\frac{1}{2} e^{s \eta (\frac{s\eta}{2} -I_d)} \mathrm{erfc}\left[\frac{K_c-I_d+s\eta_c}{\sqrt{2}}\right] 
+\nonumber\\ 
&\frac{e^{-\frac{{I_d}^2}{2} + \frac{ I_d^2}{ 2+\frac{4 s \eta_c}{K_c}}}\left(
\mathrm{erf}\left[ \frac{I_d}{\sqrt{2+\frac{4 s \eta_c}{K_c}}}\right]
+ \mathrm{erf}\left[ \frac{I_d-K_c-2 s \eta_c}{\sqrt{2+\frac{4 s \eta_c}{K_c}}}\right]\right)}{2 \sqrt{2s+K_c/\eta_c}}. 
\end{align}
The Laplace transform of the harvested power from wind energy can be given as follows:
\begin{align*}
&\mathcal{L}_{P_w}(s)
=\sum\limits_{n=0}^\infty \frac{(-s a c^3 )^n}{n!} \left(\frac{3n}{k}\right)!.
\end{align*}
\end{theorem}
\begin{proof}
See {\bf Appendices~C and D.}
\end{proof}
The Laplace Transform of the total power harvested from solar and wind sources can be given as follows:
\begin{equation}\label{combined}
\mathcal{L}_T(s)=\mathcal{L}_{P_w}(s) \mathcal{L}_P(s).
\end{equation}
The characteristic function $\phi_T(s)$ of the total power harvested from solar and wind sources can then be derived  by substituting $s=j \omega$ in the expressions derived in {\bf Theorem~3} and then substituting them into \eqref{combined}. The energy outage $E_{\mathrm{out}}$ can then be given by applying the Gil-Pelaez inversion theorem as:
\begin{align}
&E_{\mathrm{out}}=
\frac{1}{2}+\frac{1}{ \pi }\int_0^\infty \frac{\mathrm{Im}[\mathcal{\phi}_{T} (j\omega)) e^{-j \theta \omega}]}{\omega} d \omega,
\end{align}
where $\mathrm{Im}[\cdot]$ represents the imaginary operator.

\section{SNR Outage Analysis}
When the signal-to-noise ratio (SNR) of a ground cellular user falls below a certain target threshold $\mathrm{SNR_{\mathrm{th}}}$, the user is said to be in $\mathrm{SNR}$ outage.
Mathematically, the SNR outage can be defined as follows:
\begin{equation}
S_{\mathrm{out}}=\mathbb{P}(\mathrm{SNR}\leq \mathrm{SNR_{\mathrm{th}}}),
\end{equation}
where $\mathrm{SNR}_{\mathrm{th}}$ is the SNR threshold for a specific ground user so that the user can achieve minimum data rate as  \begin{equation}
    R_{\mathrm{th}}= \frac{T_b-T_f}{T_b}\mathrm{log}_2(1+\mathrm{SNR_{\mathrm{th}}}).
\end{equation} 
Subsequently,  $\mathrm{SNR}_{\mathrm{th}}$ can be calculated as $\mathrm{SNR}_{\mathrm{th}} = 2^{\frac{T_b R_{\mathrm{th}}}{T_b-T_f}}-1$. From the path-loss model in Section~II,  $\mathrm{SNR}$ is given  as:
\begin{equation}
\mathrm{SNR}=\frac{L(h_d,r_i) P_d \chi}{N_0},
\end{equation}
{where $\chi$ denotes Gamma distributed channel fading. We consider Nakagami-m fading since it provides a generalized model that mimics various fading environments\cite{downlinkcoverage}.} Thus, conditioned on the distance of the ground user from the UAV,  we can derive the SNR outage as follows:
\begin{align}\label{out1}
\begin{split}
S_{\mathrm{out}}(r) &=\mathbb{E}_r\left[\mathbb{P}\left(\mathrm{SNR} \leq \mathrm{SNR}_{\mathrm{th}}\right)\right],\\&= \mathbb{E}_r\left[\mathbb{P}\left(\frac{L(h_d,r_i) P_d \chi}{N_0} \leq \mathrm{SNR}_{\mathrm{th}}\right)\right],\\&=
\mathbb{E}_r\left[\mathbb{P}\left(\chi \leq \frac{\mathrm{SNR}_{\mathrm{th}}N_0}{L(h_d,r_i) P_d} \right)\right],
\\&=\mathbb{E}_r\left[\frac{\gamma(m,\frac{\mathrm{SNR}_{\mathrm{th}}N_0}{\Theta L(h_d,r_i) P_d})}{\Gamma(m)}\right],
\end{split}
\end{align} 
where $m$ and $\Theta$ are shape and scale parameters of Gamma distribution, respectively.
For Rayleigh fading channels, we can simplify the SNR outage probability as follows:
\begin{align}\label{out2}
S_{\mathrm{out}}(r)=1-e^{-\frac{\mathrm{SNR}_{\mathrm{th}}N_0}{L(h_d,r_i) P_d}},
\end{align} 
where $L(h_d,r_i)$ (in dB) can be expressed  as follows:
\begin{equation}
L(h_d,r_i)=\frac{c^2}{y(4 \pi f_c)^2 (h_d^2+r_i^2)}  \left(\frac{y}{x}\right)^{P_{\mathrm{LOS}}(h_d,r_i)},
\end{equation}
and $x= 10^{ \frac{\eta_{\mathrm{LOS}}} {10}}$, 
$y= 10^{ \frac{\eta_{\mathrm{NLOS}}} {10}}$, and $P_{\mathrm{LOS}}(h_d,r_i)=({1+a \mathrm{exp}(-b(\mathrm{arctan}(\frac{h_d}{r_i})-a))})^{-1}$.
Here $\mathrm{arctan}\left(\frac{h_d}{r_i}\right)$ is the elevation angle between the UAV and the served user (in degrees).  Also, $a$ and $b$ are constant values based on the choice of the urban environment. 

We now consider the problem of $\mathrm{SNR}$ outage minimization to optimize the transmit power of the UAV $P_d$. The problem can be formulated as follows:
\begin{equation*}
\begin{aligned}
& \underset{P_d}{\text{minimize}}
& & \frac{\gamma(m,\frac{\mathrm{SNR}_{\mathrm{th}}N_0}{\Theta L(h_d,r_i) P_d})}{\Gamma(m)},\\
& \text{subject to}
& & P_d \leq \frac{(P_b T_b+H T_f) - (P_{\mathrm{hov}}+\gamma_d)T_f}{T_b-T_f},
\\
& \text{}
& & T_f<T_b.
\end{aligned}
\end{equation*}
The first constraint imposes a restriction on the transmit power of the UAV. We know that, the total available UAV energy is the sum of battery energy and harvested energy in time $T_f$, i.e. $P_b T_b+ H T_f$. And the total available UAV energy should be greater than the energy required during the flight as well as the transmission time $(P_d+P_{\mathrm{hov}})(T_b-T_f)+(P_{\mathrm{hov}}+\gamma_d)T_f$. As a result, we have an upper bound on the transmit power of the UAV. The second constraint denotes the fact that a UAV can not have a flight time which exceeds the observation period $T_b$. Since the lower incomplete Gamma function is monotonically decreasing function with respect to increasing $P_d$, the optimal solution for $P_d$ will lie at the boundary, i.e.
\begin{equation}
    P_d^*= \frac{(P_b T_b+H T_f) - (P_{\mathrm{hov}}+\gamma_d)T_f}{T_b-T_f}.
\end{equation}
The first constraint imposes a restriction on the transmit power of the UAV. We know that, the total available UAV energy is the sum of battery energy and harvested energy in time $T_f$, i.e. $P_b T_b+ H T_f$. And the total available UAV energy should be greater than the energy required during the flight as well as the transmission time $(P_d+P_{\mathrm{hov}})(T_b-T_f)+(P_{\mathrm{hov}}+\gamma_d)T_f$. As a result, we have an upper bound on the transmit power of the UAV. The second constraint denotes the fact that a UAV can not have a flight time which exceeds the observation period $T_b$. Since the lower incomplete Gamma function is monotonically decreasing function with respect to increasing $P_d$, the optimal solution for $P_d$ will lie at the boundary, i.e.
\begin{equation}
    P_d^*= \frac{(P_b T_b+H T_f) - (P_{\mathrm{hov}}+\gamma_d)T_f}{T_b-T_f}.
\end{equation}

We now consider the problem of $\mathrm{SNR}$ outage minimization to optimize the flight time of the UAV $T_f$. The problem can be formulated as follows:
\begin{equation*}
\begin{aligned}
& \underset{T_f}{\text{minimize}}
& & \frac{\gamma\left(m,\frac{(2^{\frac{T_b R_{\mathrm{th}}}{T_b-T_f}}-1 )N_0}{\Theta L(h_d,r_i) P_d}\right)}{\Gamma(m)},\\
& \text{subject to}
& & P_d \leq \frac{(P_b T_b+H T_f) - (P_{\mathrm{hov}}+\gamma_d)T_f}{T_b-T_f},
\\
& \text{}
& & T_f<T_b.
\end{aligned}
\end{equation*}
 Since the lower incomplete Gamma function is monotonically increasing with respect to increasing $T_f$, the optimal solution for $T_f$ will be based upon the minimum value of the $T_f$ obtained from constraints, i.e.
\begin{equation}
    T_f^*= \mathrm{min}\left(T_b, \frac{(P_d-P_b) T_b}{H-(P_{\mathrm{hov}}+\gamma_d)+P_d}\right).
\end{equation}
The distance between the origin  and an arbitrary user whose location is uniformly distributed within a circular region of radius $R$ has the PDF given by: 
\begin{equation}
f_R(r)=\frac{2r}{R^2}, 
\end{equation}
where $0 \leq r \leq R$. By averaging over the distribution of $r$, the SNR outage of a given  user can thus be calculated as follows:
\begin{align}
&S_{\mathrm{out}}= \int\limits_0^{R} S_{\mathrm{out}}(r)\frac{2r}{R^2} \mathrm{dr},
\end{align}
in which $S_\mathrm{out}(r)$ is given in \eqref{out1} and \eqref{out2} for Gamma fading and Rayleigh fading channels, respectively.
For Rayleigh fading channels, the SNR outage can be computed as follows:
\begin{align}
=1-\int\limits_0^{R}  e^{- C (h_d^2+r^2) z^{P_{\mathrm{LOS}}(h_d,r)} }\frac{2r}{R^2} \mathrm{dr},
\end{align}
where $C=\frac{SNR_{\mathrm{th}} N_0 y(4 \pi f_c)^2}{c^2  P_d}$ and $z=\frac{x}{y}$. The integral can be computed numerically using standard mathematical software packages such as \texttt{Mathematica, Matlab and Maple}.

\section{Moments of the Harvested Power and Their Applications}

In this section, we first derive the moments of the harvested solar and wind power. Based on these moments, we characterize important metrics such as eventual energy outage probability, probability of charging the UAV battery in a finite time, average battery charging time, probability of eventual energy outage, and energy outage probability considering that  the energy arrival  follows a discrete stochastic process where both the amount of energy received at a certain time as well as  the inter-arrival times are random.

\subsection{Moments of the Harvested Solar Power}
The $i$-th moment of the harvested solar power can be calculated by applying the definition  $\mathbb{E}[P^i]=\int\limits_{0}^\infty p^i f_P(p) dp$. The closed-form expressions for the $i$-th moment of distribution of harvested power from solar energy $M_i$ can be derived as follows:
\begin{align}
&\mathbb{E}[P^i] =\frac{ \int\limits_0^{\eta_c K_c} p^{i-0.5}e^{ -\frac{\left(\sqrt{\frac{K_c p}{\eta_c}}-I_d\right)^2}{2}}}{2 \sqrt{2 \pi \eta_c/K_c}}+\frac{\int\limits_{\eta_c K_c}^\infty p^i e^{-\frac{\left(\frac{p}{\eta_c}-I_d\right)^2}{2}}  }{\eta_c \sqrt{2 \pi}}, 
\nonumber\\
&\stackrel{(a)}{=}\underbrace{\frac{\int\limits_{-I_d}^{K_c-I_d} (x+I_d)^{2i} e^{-\frac{x^2}{2}} dx }{\left(\frac{K_c}{\eta_c}\right)^i\sqrt{2\pi}}}_{{\bf I_1}}  + \underbrace{\frac{\int\limits_{K_c-I_d}^\infty (y+I_d)^i \eta_c^i e^{-\frac{y^2}{2}} dy}{\sqrt{2\pi}}}_{{\bf I_2}},
\label{Lap11}
\end{align}
where (a) results from substituting $\sqrt{\frac{K_c p}{\eta_c}}-I_d=x$ and $\frac{p}{\eta_c}-I_d=y$ in the first  and second integral, respectively. The integral in ${\bf I_1}$ can be solved in closed-form as follows:
\begin{align*}
{\bf I_1}=\frac{\eta_c^4 e^{-\frac{I_d^2}{2}} \left(
	2m^\prime -2 e^{I_d K_c -\frac{K_c^2}{2}} A+
	e^{\frac{I_d^2}{2}}d\sqrt{2 \pi} \nu
	\right)}{K_c^4 2\sqrt{2 \pi}},  
\end{align*}
where $m^\prime = I_d(279+185I_d^2+27I_d^4+I_d^6)$, $a=(7+I_d^2)(15+18I_d^2+I_d^4)$, $b=I_d(87+22I_d^2+I_d^4)$, $c=(35+18I_d^2+I_d^4)$, $d=105+420 I_d^2+210I_d^4+28I_d^6+I_d^8$, $A=m^\prime+aK_c+b K_c^2+   c K_c^3+I_d(13+I_d^2)K_c^4+(7+I_d^2)K_c^5+I_d K_c^6+K_c^7 $.

Similarly, the integral in ${\bf I_2}$ can be solved as follows:
\begin{align*}
{\bf I_2}=\frac{\eta_c^2}{\sqrt{2\pi}}(I_d + K_c)e^{-\frac{(I_d-K_c)^2}{2}}+\frac{\eta_c^2}{2}(1+I_d^2) A^\prime,
\end{align*}
where $A^\prime= 1+\mathrm{erf} \left(\frac{I_d-K_c}{\sqrt{2} \mathrm{Sign}(I_d-K_c)} \right)\mathrm{sign}[I_d-K_c]^3 $.

Subsequently, the first moment can be simplified in the closed-form as in the following Lemma.
\begin{Corollary}[Mean Harvested Solar Power]
The closed-form expression for the first moment of harvested solar power $\mathbb{E}[P]$ can be obtained by substituting $i=1$ in \eqref{Lap11} as shown below: 
	\begin{align*}
	M_1=\frac{1}{2} \eta_c I_d \tau-\frac{I_d \eta_c ( e^{-\frac{(I_d-K_c)^2}{2}}-e^{-\frac{I_d^2}{2}} )}{\sqrt{2 \pi} K_c}  -\frac{1}{2} \frac{\eta_c}{K_c}(1+I_d^2)\nu,
	\end{align*}
	where $\nu=\tau-1-\mathrm{erf}\left(\frac{I_d}{\sqrt{2}}\right)$ and $\tau= 1+ \mathrm{erf}\left(\frac{I_d-K_c}{\sqrt{2}}\right)$.
\end{Corollary}	

\begin{Corollary}[Second Moment of Harvested Solar Power]
	The second moment of harvested solar power can be obtained in closed-form as follows: 
	\begin{align*}
	&M_2=
	\frac{\eta_c^2 e^{-\frac{I_d^2}{2}}I_d (5+I_d^2)}{K_c^2\sqrt{2\pi}}  +\frac{\eta_c^2 (1+I_d^2)\tau}{2}-\frac{\eta_c^2(3+6I_d+I_d^4)\nu }{2K_c^2} \\&+\frac{\eta_c^2}{\sqrt{2\pi}} e^{-\frac{(I_d-K_c)^2}{2}} \left(I_d-\frac{I_d^3}{K_c^2}-\frac{3}{K_c}-\frac{I_d^2}{K_c}-\frac{5I_d}{K_c^2}-1\right),
	\end{align*}
\end{Corollary}	
where $\nu=\tau-1-\mathrm{erf}\left(\frac{I_d}{\sqrt{2}}\right)$ and $\tau= 1+ \mathrm{erf}\left(\frac{I_d-K_c}{\sqrt{2}}\right)$.

\subsection{Moments of the Harvested Wind Power}
The $i$-th moment of the harvested wind  power can be calculated by applying the definition  $\mathbb{E}[P_w^i]=\int\limits_{0}^\infty p_w^i f_{P_w}(p_w) dp_w$. The closed-form expressions for the $i$-th moment of wind distribution $M_{wi}$ can be derived as follows:
\begin{align*}
M_{w_i}=\frac{k}{3 a^{k/3} c^k} \int\limits_{a{V_{ci}}^3}^{a{V_{r}}^3}   x^{\frac{k}{3}+i-1} \mathrm{exp} \left(-\frac{x^{\frac{k}{3}}}{a^{\frac{k}{3} }c^k} \right) dx.
\end{align*}
This can be solved using the identity 
$
\int\limits_0^u x^m e^{-\beta x^n} dx= \frac{\gamma(v, \beta u^n)}{n \beta^v}, v=\frac{m+1}{n}
$~\cite{ebook}.
where $u>0$, Re $v>0$, Re $n>0$ and Re $\beta>0$. Hence, we can solve $M_{w_i}$ considering $m=\frac{k}{3}+i-1, n=\frac{k}{3}$ and $\beta=a^{-\frac{k}{3} }c^{-k}$. Thus $v=1+\frac{3i}{k}+1$. We also know that $\int\limits_a^b f(x) dx= \int\limits_0^{b^\prime} f(x) dx - \int\limits_0^{a^\prime} f(x) dx$.
Here $b^\prime= a {V_{r}}^3$ and $a^\prime=a{V_{ci}}^3$. Finally after simplification $M_{w_i}$ can be expressed as:
\begin{equation}
M_{w_i}=a^i c^{3i} \left[ \gamma \left( 1+\frac{3i}{k}, \left(\frac{V_r}{c} \right)^k \right) - \gamma \left( 1+\frac{3i}{k}, \left(\frac{V_{ci}}{c} \right)^k \right) \right].
\end{equation}

Subsequently, the first and second moments can be simplified in the closed-form as in the following Lemma.
\begin{Corollary}[First and Second Moments of Harvested Wind Power]
	The first and second moment of harvested wind power can be obtained in closed-form, respectively, as follows: 
	\begin{equation*}
	M_{w_1}=a c^3 \left[ \gamma \left( 1+\frac{3}{k}, \left(\frac{V_r}{c} \right)^k \right) - \gamma \left( 1+\frac{3}{k}, \left(\frac{V_{ci}}{c} \right)^k \right) \right],
	\end{equation*}
	\begin{equation*}
	M_{w_2}=a^2 c^6 \left[ \gamma \left( 1+\frac{6}{k}, \left(\frac{V_r}{c} \right)^k \right) - \gamma \left( 1+\frac{6}{k}, \left(\frac{V_{ci}}{c} \right)^k \right) \right].
	\end{equation*}
\end{Corollary}	

\subsection{Applications of the Moments of the Harvested Power}

To compute the probability of charging the UAV battery in a finite time, average battery charging time, eventual energy outage probability, and average energy outage, we make the following assumptions. The energy arrivals occur in bursts (i.e. as ``energy packets") and the size of the energy packets be $X_1, X_2,\cdots, X_i$ arriving at time $t_1,t_2,\cdots, t_i$, respectively. Here, packet size $X$ is a random variable which follows the distribution of solar or wind harvested power. Let $A_1, A_2, \cdots, A_i$ denote the inter-arrival times of the energy packets such that $t_n= A_0+A_1+\cdots+A_n$. Due to the memory-less property of the Poisson process, each $A_i$ is exponentially distributed with parameter $\lambda$. The mean and variance of $A$ are finite. Since $X$ and $A$ are independent random variables, the accumulated energy in the battery at any time $t$ is given as:
\begin{equation}
U(t)=\sum_{i=1}^{N_A(t)} X_i,
\end{equation}
where $N_A(t)$ denotes the number of packets arrived in time $t$.

To connect the above assumptions with the system model assumed before, note that, in the previous sections,  the  harvested  power $H$ has been a random variable whose value varies  across the flight intervals and whose distribution is given using  {\bf Theorem~1} and {\bf Theorem~2} for solar power ($H = P$)  and wind power ($H = P_w$), respectively.  Here, in this subsection, we consider that the harvested power $H$ is equivalent to the harvested energy per unit time, i.e. $X=H t$, where $X$ is the energy packet size. Assuming $t=1$ for simplicity, the energy packet size is a random variable whose PDF can be given using {\bf Theorem~1} and {\bf Theorem~2} for solar and wind, respectively. For other values of $t$, the distribution of $X$ can be obtained by scaling the distribution of $H$ using single variable transformation. Furthermore, the energy packet arrival follows a Poisson process, i.e. the inter-arrival time $A_i$ is exponentially distributed  with arrival rate $\lambda$.

\subsubsection{Probability of Charging UAV Battery During Flight Time} 

The time needed to charge the battery up to a desired level, $u_0>0$ can be formulated as a first passage time problem and can be derived as follows:
$\tau(u_0)=\mathrm{inf}_t\{t:U(t)>u_0\}$, where $\tau{(u_0)}$ denotes the battery recharge time. Since the charging process is a pure jump process for an ideal battery,  the event $U(t)>u_0$ is equivalent to the event  $\tau(u_0)<t$  which results in $\mathbb{P}(U(t)>u_0)=\mathbb{P}(\tau(u_0)<t)$. Now to find the distribution of $U(t)$, we condition on the number of packets and represent the sum of random variables $X_i$ by convolving their distributions. Finally, we average over the distribution  of the number of packets arrived in time $t$ to obtain:
\begin{equation}
\mathbb{P}(U(t)\leq u_0)=e^{-\lambda t}\sum\limits_{n=0}^\infty \frac{(\lambda t)^n}{n!}F_X^{(n)}(u_0),
\end{equation}
where $F_X^{(n)}(x)$ denotes the $n$-fold convolution of $F_X(x)$ where
$F_X^{(i)}(x)=\int F_X^{(i-1)}(x-t) dF(t)$.
$F_X^{(0)}(x)$ is the unit step function at the origin. The probability of UAV battery charging in time $T_f$ can thus be given as follows~\cite{sudarshan_journal}:
\begin{equation}
\mathbb{P}(\tau(u_0) \leq T_f)=1-e^{-\lambda T_f}\sum\limits_{n=0}^\infty \frac{(\lambda T_f)^n}{n!}F_X^{(n)}(u_0),
\end{equation} 
where  $F_X^{(n)}(u_0)$ can be approximated using central limit theorem as $   \Phi \left( \frac{u_0-n \bar{X}}{\sigma_X \sqrt{n}}\right)$ \cite{sudarshan_journal}, where $\Phi(.)$ denotes the CDF of the standard normal distribution.

\subsubsection{Average Charging Time of UAV Battery}
Subsequently, the expected charging time for UAV battery can be derived as:
\begin{align*}
\mathbb{E}[\tau(u_0)]=&\int\limits_0^{T_f} \mathbb{P}(\tau(u_0) \geq t)dt,\\=&\sum\limits_{n=0}^\infty \Phi \left( \frac{u_0-n \bar{X}}{\sigma_X \sqrt{n}}\right) \frac{1}{n!} \int\limits_0^{T_f} (\lambda t)^n e^{-\lambda t} dt.
\end{align*}
Finally, the average charging time can be given as:
\begin{align*}
\mathbb{E}[\tau(u_0)]=&\sum\limits_{n=0}^\infty \Phi \left( \frac{u_0-n \bar{X}}{\sigma_X \sqrt{n}}\right) \frac{1}{n!} \int\limits_0^{T_f} (\lambda t)^n e^{-\lambda t} dt,\\
=&\frac{\Gamma_l(1+n, T_f \lambda)}{n! \lambda}\sum\limits_{n=0}^\infty \Phi \left( \frac{u_0-n \bar{X}}{\sigma_X \sqrt{n}}\right),
\end{align*}
where  $\bar{X}$ and $\sigma_X$ denote the mean and standard deviation of the solar and wind energy, $u_0$ is the required level of energy up to which we need to charge the battery, $\tau(u_0)$ is the time required to reach the level of energy $u$, $\lambda$ denotes the energy arrival rate, and $n$ is the total number of energy packets arrived. 

\subsubsection{Eventual Energy Outage} The {eventual energy outage} $\phi(u_0)$  is defined as the energy outage that occurs within a finite amount of time while the initial battery energy is $u_0$ and the power consumed is either $P_{\mathrm{hov}} + \gamma_d$ during the flight and $P_{\mathrm{hov}}+P_d$ during the hover time. If the energy packet arrival can be modeled as a Poisson process, the energy surplus at any time $t$ can be modeled as follows:
\small
\begin{equation*}
    U_1(t)=
    \begin{cases}
    u_0-(P_{\mathrm{hov}} + \gamma_d)t+\sum_{i=1}^{N(t)} X_i, & \\
             \quad \quad \quad \quad \mbox{for }t <T_f & \\
    u_0-P_{\mathrm{hov}}t -P_d (t-T_f)-\gamma_d T_f+\sum_{i=1}^{N(T_f)} X_i, & \\
         \quad \quad \quad \quad \mbox{for }t>T_f & \\
    \end{cases}.
\end{equation*}
\normalsize
Here $u_0 \geq 0$ denotes the initial battery energy and  $N(t)$ is the total number of energy packets harvested within time $t$. The energy outage that occurs within finite amount of time making the battery energy to zero is known as eventual energy outage $\phi(u_0)$. 
The first time when battery energy goes to zero is denoted as first time to outage $\tau$. $\tau=\mathrm{inf}\{t>0: U_1(t) \leq 0, U_1(0)=u_0\}$. If $\tau$ is less than $T_f$ or if $T_f \leq \tau \leq T_b$, we say that eventual outage has occurred. Thus, the probability of eventual outage within duration $T_f$ and $T_b-T_f$ can be modeled, respectively, as follows:

\begin{equation}
    \phi(u_0)=
    P(\tau < T_f|u_0, P_{\mathrm{hov}} + \gamma_d), 
\end{equation}
\begin{equation}
\phi(u_0)=
    P(T_f < \tau < T_b|u_0+P_d T_f -\gamma_d T_f,P_{\mathrm{hov}} +P_d). 
\end{equation}
\normalsize
The eventual energy outage probability $\phi(u_0)$ within duration $T_f$ and $T_b-T_f$ can be given using the approach described in \cite{sudarshan_journal}, respectively, as follows:
\small
\begin{equation}
    \phi(u_0)=
    \left(1-\frac{r^* (P_{\mathrm{hov}}+\gamma_d)}{\lambda}\right)e^{- r^* u_0},
\end{equation}
\begin{equation}
\phi(u_0)=
    \left(1-\frac{r^* (P_{\mathrm{hov}} +P_d)}{\lambda}\right)e^{- r^* (u_0+P_d T_f -\gamma_d T_f}),
\end{equation}
\normalsize
where $r^*$ denotes the adjustment coefficient. The value $r^* \neq 0$ is said to be the adjustment coefficient of $X$ if $\mathbb{E}[\mathrm{exp}(r^* X)]=\int e^{r^* x} dF_X=1$. The adjustment coefficient must satisfy the condition
$\mathbb{E}[\mathrm{exp}(r^* X)]=\mathcal{M}_X(r^*)$. Considering Poisson energy arrival, we have: 
\begin{equation*}
1-\frac{p r^*}{\lambda}=\mathcal{M}_X(-r^*),
\end{equation*}
where $p=P_{\mathrm{hov}}+\gamma_d$ during the flight time and $p=P_{\mathrm{hov}}+P_d$ during the transmission time.
A simple expression for $r^*$  can be obtained by making a formal power expansion of $\mathcal{M}_X(-r)$ in terms of moments of $X$ up to second order term as $\mathcal{M}_X(-r) \approx 1- \bar{X} r +\frac{\bar{X^2}}{2}r^2$. As such, a simple generalized approximation of $r^*$  can be given as \cite{sudarshan_journal}:
\begin{equation}
r^* \approx \frac{2p}{\lambda \bar{X^2}}\left(\frac{\lambda \bar{X}}{p}-1\right),
\end{equation}
where $\bar{X}$ is given in \textbf{Corollary~1} and \textbf{Corollary~3} for solar and wind power, respectively.
\subsubsection{Energy Outage}
From Proposition 10 of \cite{sudarshan_journal}, we know that for any harvest store consume system, if there exists a non-negative probability of eventual energy outage, the average energy outage can be calculated as follows:
\begin{equation}
E_{\mathrm{out}}=
\begin{cases}
1-\frac{\lambda \bar{X}}{P_{\mathrm{hov}}+\gamma_d}, &t< T_f\\
1-\frac{\lambda \bar{X}}{P_{\mathrm{hov}}+ P_d}, &t>T_f
\end{cases}.
\end{equation}

\section{Numerical Results and Discussions}
In this section, we first describe the simulation parameters and then present selected numerical and simulation results. We validate the accuracy of the derived expressions for the energy outage probability of UAV and the SNR outage of a ground user. Also, we extract design insights related to the impact of  key UAV parameters such as flight time, transmit power, altitude, and time of the day on the energy outage of UAV and SNR outage of the ground user.

\subsection{Simulation Parameters}

Unless stated otherwise, we use the simulation parameters as listed herein. We consider the maximum sunlight intensity  $I_{\mathrm{max}}=2000$. The threshold for radiation intensity $I$ is $K_c=150$. The efficiency of the solar cell beyond $K_c$ is $\eta_c=0.02$. The random attenuation amount $\Delta I$ follows a normal distribution with mean $\mu =0$ and variance $\sigma^2=1$. The mass of a UAV is taken as  $0.75~\mathrm{kg}$. The number of propellers and propeller radius are, respectively, $4$ and $0.2 $m. The air density is $\rho = 1.225~\mathrm{kg}/m^3$ and $g=9.8$~m/s$^2$ is the gravity of the earth. The activation energy of the UAV is $\gamma_d=2.9$~W. We consider the initial battery power $P_b=2+P_{\mathrm{hov}}+\gamma_d$. Also, $f_c=2.5 \times 10^9$ is the carrier frequency (Hz) and $c=3 \times 10^8$~m/s is the speed of light. The UAV altitude is set at $h=200$~m. The UAV power is assumed to be $P_d = 40 $ W. The UAV flies at a speed $v=10$~m/s. The block duration is set to $T_b=R_{\mathrm{max}}/v_d$, where $R_{\mathrm{max}}=200$~m. The threshold rate requirement for a cellular user is set to $R_{\mathrm{th}}=2$~bps. The S-curve parameters for the UAV are $a=12.08$ and $b=0.11$. Also, $\eta_\mathrm{LOS}=1.6$ and $\eta_ \mathrm{NLOS}=23$ (in dB) are, respectively, the losses corresponding to the LOS and non-LOS reception depending on the environment. The flight time $T_f$ is set to $0.2T_b$, unless  specified otherwise.

\subsection{Results and Discussions}

\subsubsection{Distribution of the Harvested Power}

\begin{figure}[t]
\begin{center}
\includegraphics[scale=.45]{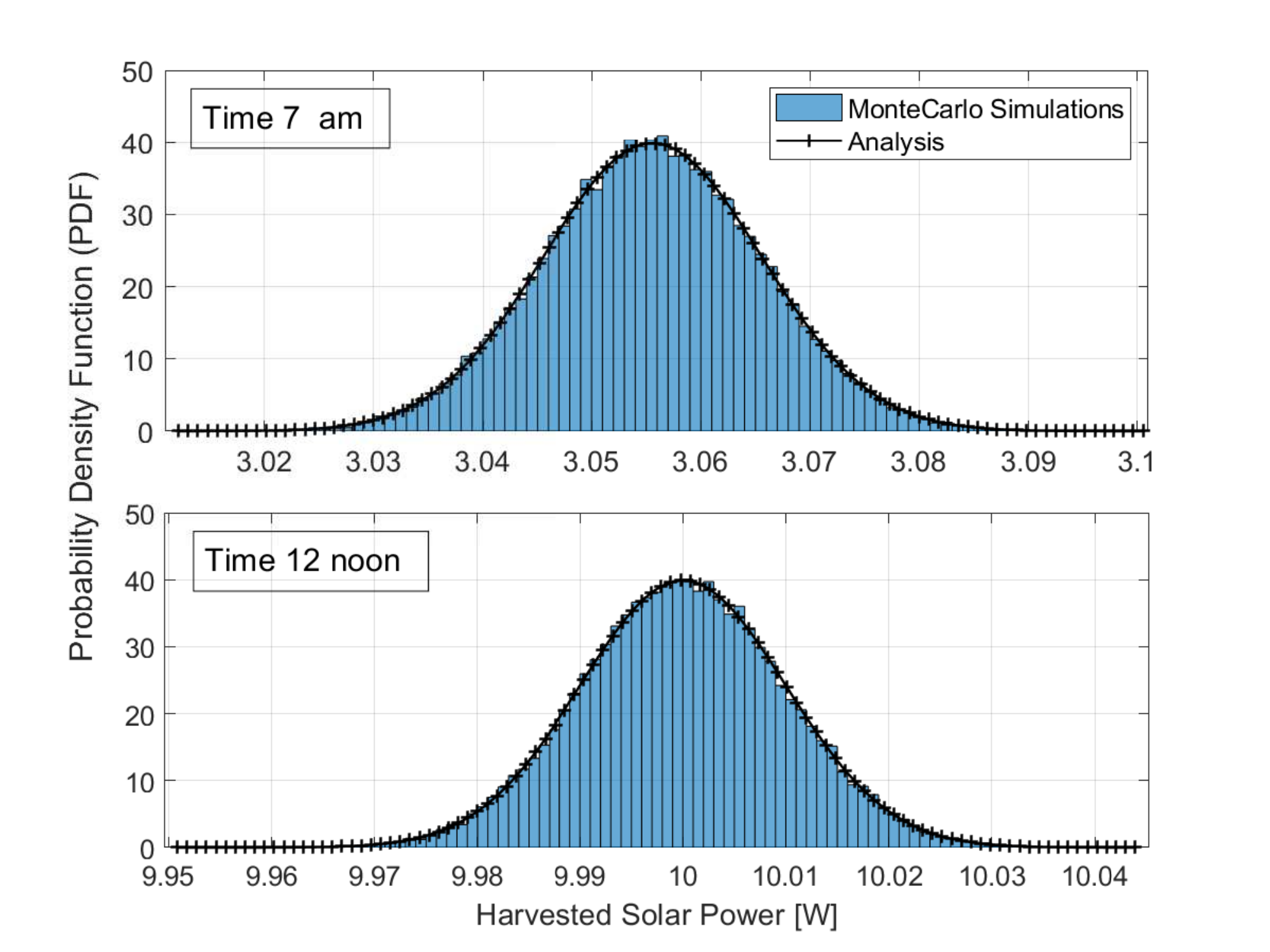}
\end{center}
\caption{Distribution of the harvested power from solar energy as a function of daytime.}
\label{fig1}
\end{figure}
\figref{fig1} depicts the level of the harvested power from solar energy at different times of the day. The Monte-Carlo simulation results match perfectly with the analytical results as we derive an exact expression for the harvested solar power. It can be seen that the harvested power is maximum at noon. Along similar lines, \figref{fig2} depicts the level of the harvested power from the wind energy. Since the PDF of wind power is a mixture of continuous and discrete random variables, we can observe two discrete probabilities when harvested wind power is zero and when the wind power becomes equal to the rated power. 

\begin{figure}[t]
\begin{center}
\includegraphics[scale=.55]{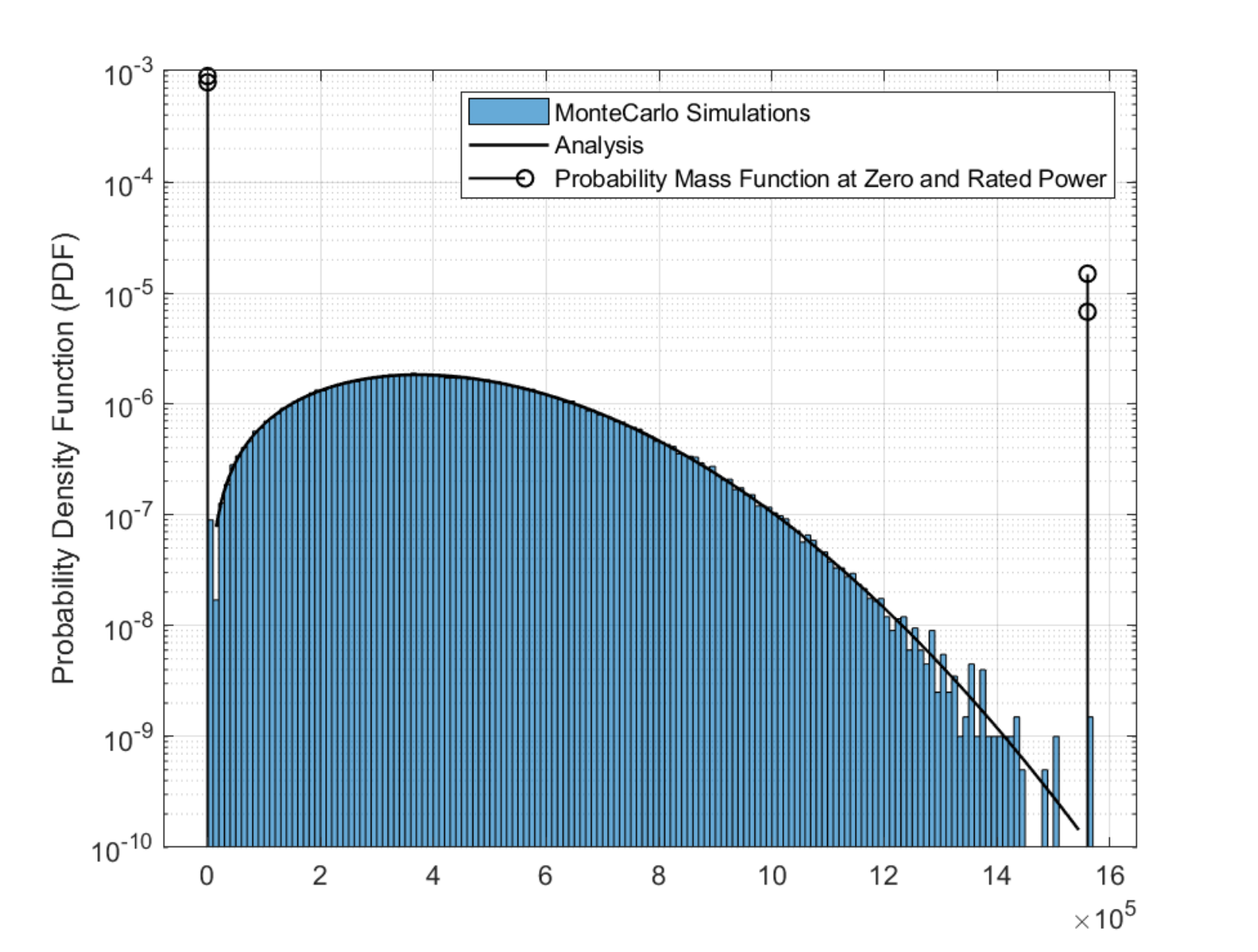}
\end{center}
\caption{Mixed probability mass/density function of the harvested power from wind energy.}
\label{fig2}
\end{figure}

\subsubsection{Impact of the Time of the Day}

\begin{figure}[t]
\begin{center}
\includegraphics[scale=.45]{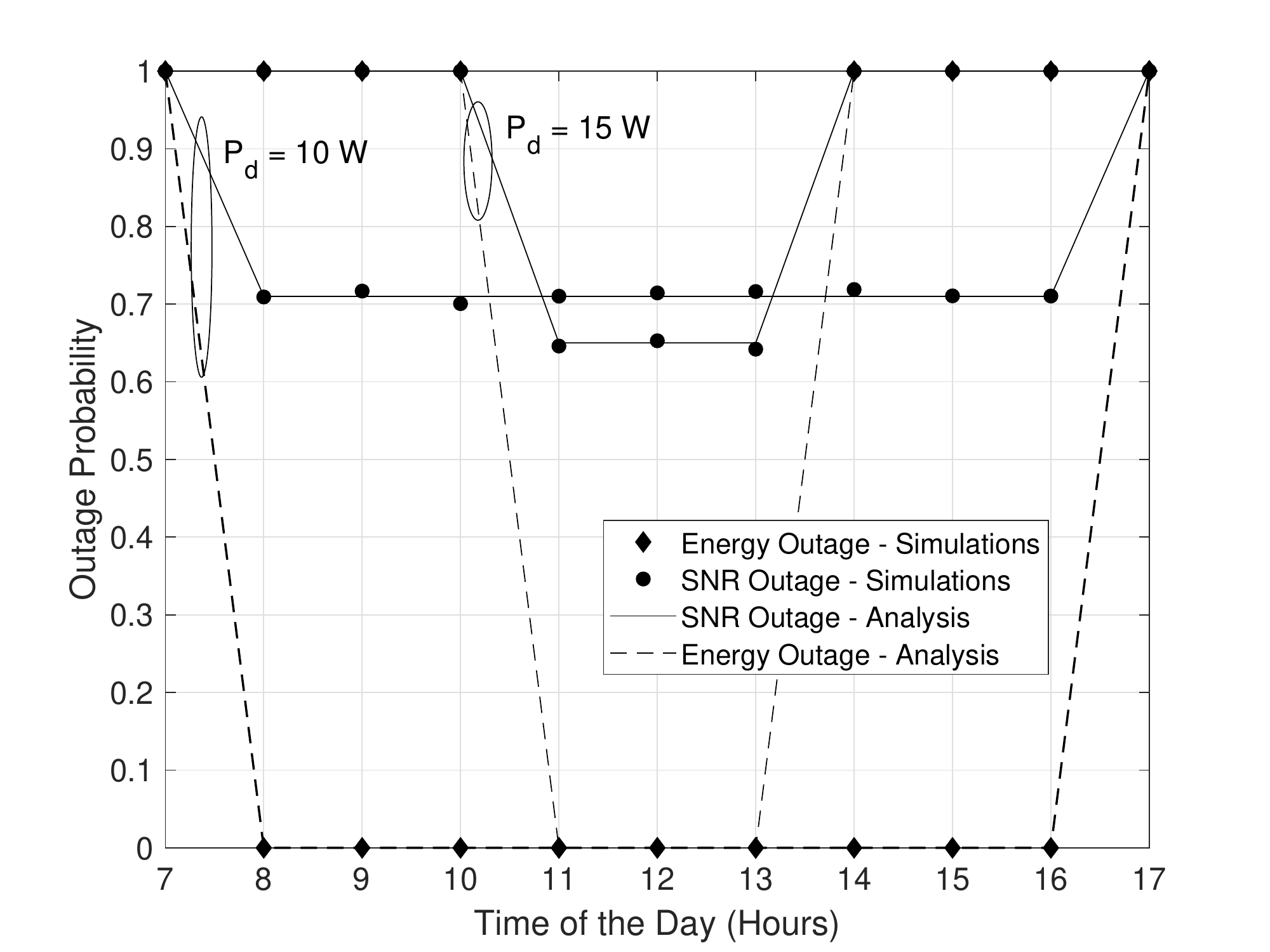}
\end{center}
\caption{SNR outage of the user and energy outage at the UAV as a function of the harvesting time. }
\label{fig4}
\end{figure}

\begin{figure}[t]
	\begin{center}
		\includegraphics[scale=.65]{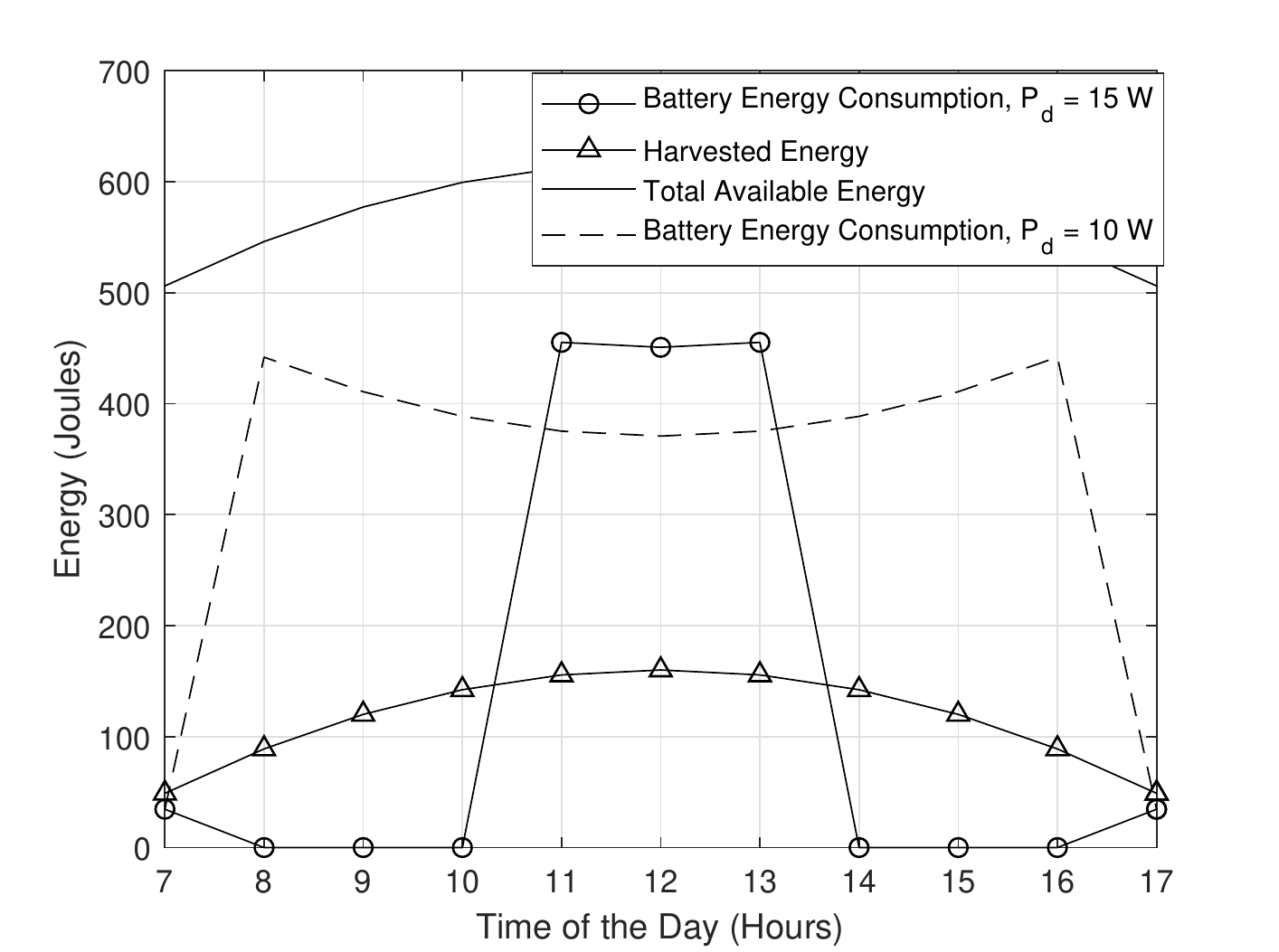}
	\end{center}
	\caption{Harvested and consumed energy of UAV as a function of the harvesting time. }
	\label{fig5}
\end{figure}

In Fig.~\ref{fig4}, we quantify the energy outage probability of the UAV and the SNR outage probability of the user  as a function of  time of the day. Our simulations match well with the derived expressions. 
Since the harvested solar energy varies significantly with the time of harvesting, the energy outage of the UAV and the SNR outage of the user vary accordingly. For example, the probability of energy outage is high in the morning due to lower harvested power and continues to increase until noon. After noon, the harvested power reduces and so the energy outage as expected. Subsequently, the SNR outage probability of a user reduces in first half of the day  and becomes fixed for a specific time duration in which energy outage does not happen. On the other hand, during the time from noon to evening, the SNR outage again increases due to frequent energy outages resulting in  transmission failures to the ground user. In addition, it can be observed that  reducing the transmission power from $P_d=15$~W to $P_d=10$~W will reduce the time duration in which the energy and SNR outages are expected. That is, with $P_d=10$~W the energy outage is not expected from 8:00 to 16:00 hours compared to when $P_d=10$~W where energy outage is expected from 7:00 to 10:00 and 14:00 to 17:00 hours.

In Fig.~\ref{fig5}, we analyze the consumed energy from the battery as a function of the  time. That is, the more harvested energy we have, less consumption from the battery occurs. This is evident from the case when $P_d=10$~W. However, for both cases, we can see that the battery consumption was low initially and then raised significantly which is counter-intuitive. The reason is that when the harvested energy is below a certain threshold, the UAV cannot  support transmission and, subsequently, resumes  return flight which consumes relatively less power  compared to the hovering and transmission consumption. However, as soon as the harvested energy plus UAV battery becomes sufficient to support the transmission, the UAV transmits for $T_b-T_f$ duration resulting in the increased power consumption. For the case of $P_d=15$~W, the harvested energy level needs to be significantly  high compared to $P_d=10$~W to support transmission. Note that, the harvested energy does not vary significantly from 10 am to 2 pm, therefore, a significant power consumption reduction cannot be observed as can be seen in $P_d=10$~W scenario.

\subsubsection{Impact of the UAV Transmit Power}
\begin{figure}[t]
	\begin{center}
		\includegraphics[scale=.55]{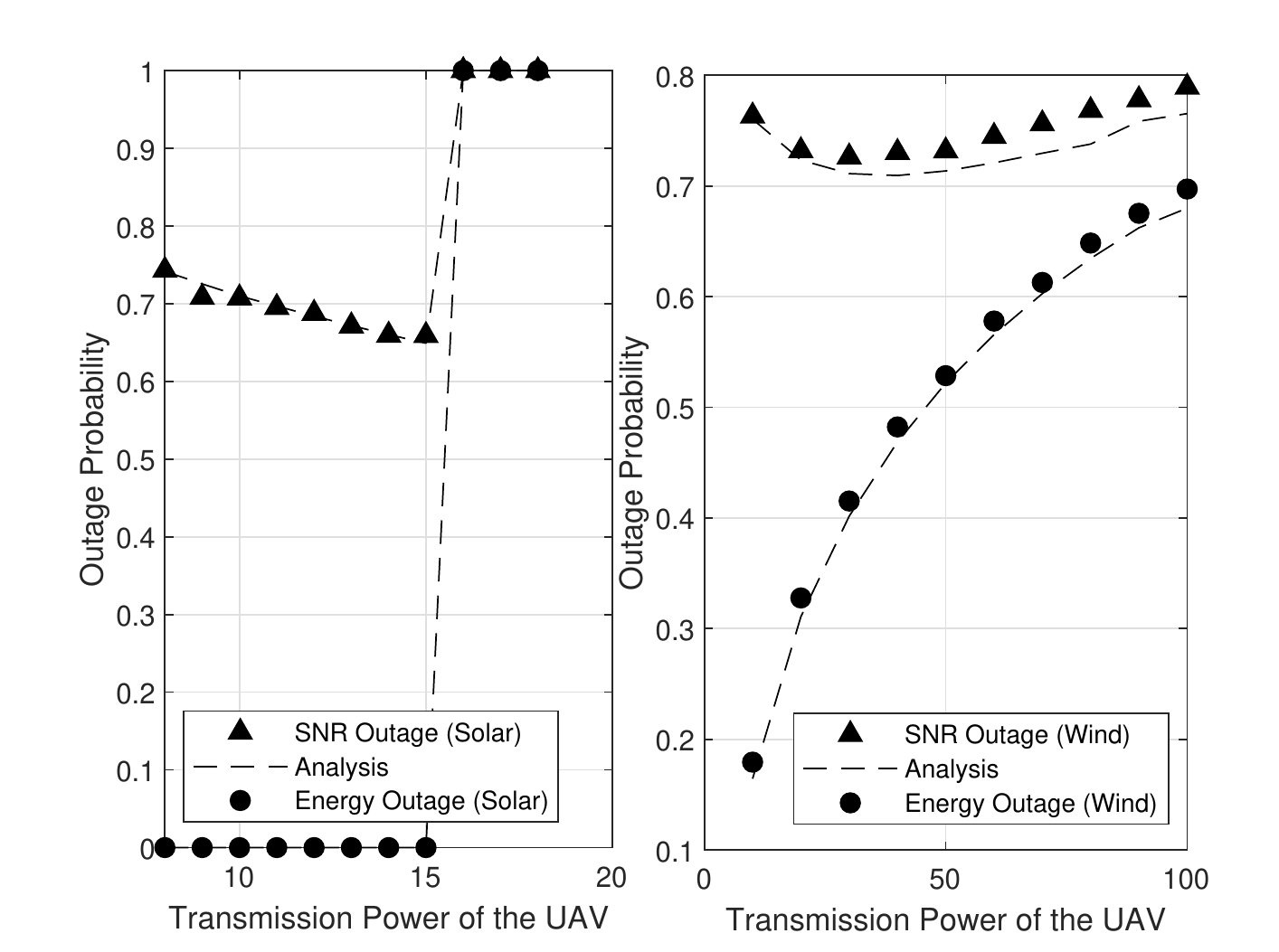}
	\end{center}
	\caption{Energy outage and rate outage as a function of the  transmission power of the UAV $P_d$ and the flight duration $T_f$. }
	\label{fig6}
\end{figure}

\figref{fig6}(a) and (b) depict the energy outage probability of the UAV and the SNR outage probability of the ground user considering both the solar and wind energy sources as a function of the transmission power of the drone, respectively. Analytical results match well with the simulations. We can observe that the increase in the transmit power increases the energy outage due to the increased energy consumption. On the other hand,  the SNR outage at UAV initially reduces due to increased SNR values; however, there is a maximum transmit power threshold beyond which the energy outage at UAV becomes more evident and in turn SNR outage happens. That is, an increase in transmit power increases the energy consumption from the UAV battery and at some point the UAV battery combined with the harvested energy will not be able to support cellular transmission and ultimately the UAV will  fly back to its charging station. Finally, both the solar and wind harvested powers show similar trends. However, the energy outage trends are sharp for solar energy compared to the wind energy. The reason is that the solar energy is nearly deterministic for a given time of the day, whereas the wind energy is random.

\begin{figure}[t]
	\begin{center}
		\includegraphics[scale=.4]{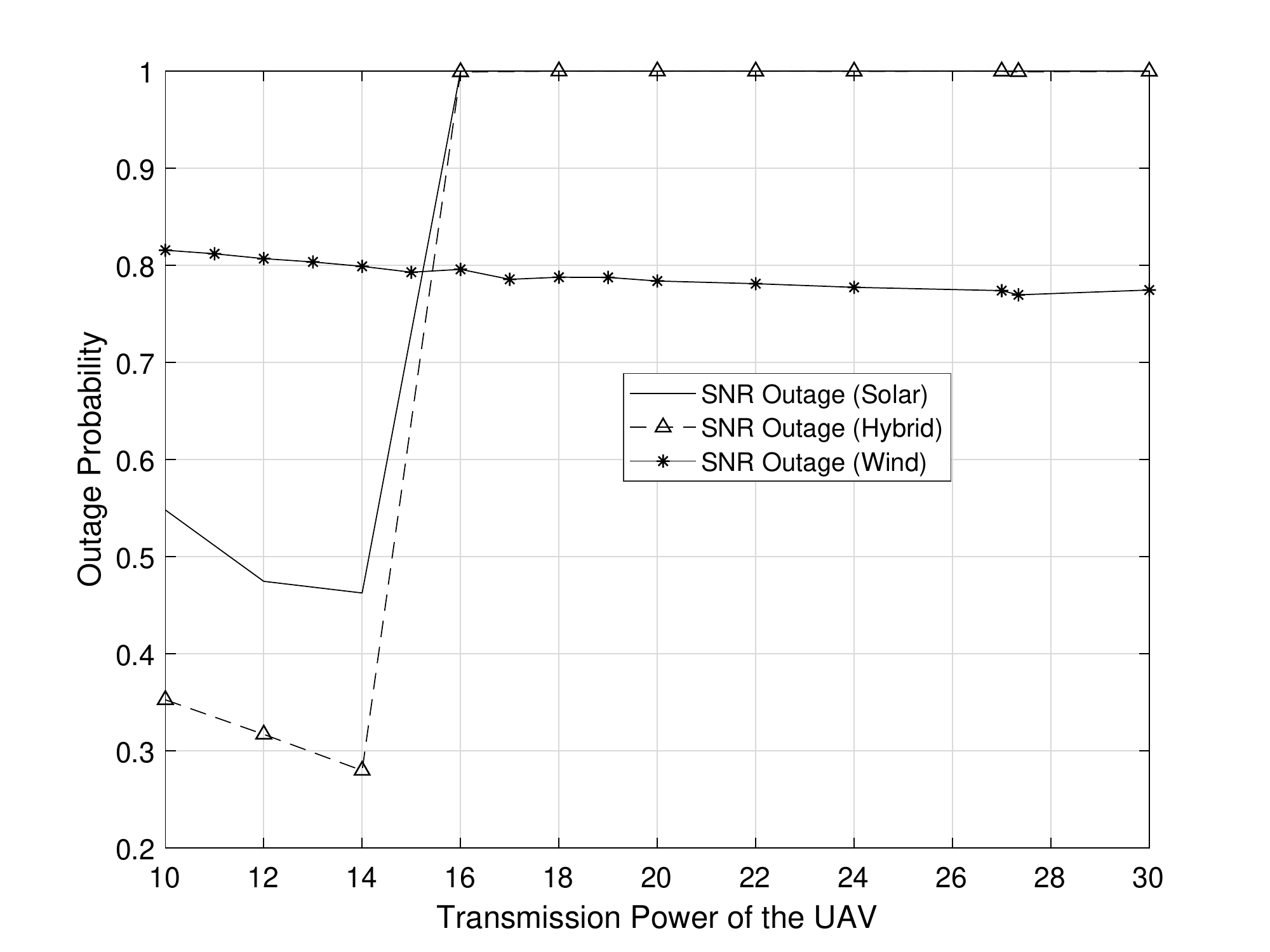}
	\end{center}
	\caption{SNR outage as a function of the  transmission power of the UAV $P_d$. }
	\label{fig6c}
\end{figure}

In \figref{fig6c}, the SNR outage probabilities are plotted against the transmit power of the UAV with three configurations.
The energy outage  for hybrid configuration follows harvested solar power as solar energy is more deterministic with higher harvested levels compared to the wind energy (especially around noon). Another observation is that with an increase in transmit power, the change in wind SNR outage is almost nearly same as the transitions in wind energy are less significant compared to the solar energy. Finally, we observe an improved SNR outage performance with the hybrid (solar and wind) harvesting model. 

In \figref{fig7} (a) and (b), the SNR outage probability for a user is plotted against the altitude of the UAV. It can be observed that an increase in the altitude of the UAV first reduces the SNR outage due to an increased probability of LOS and then increases the SNR outage due to a higher distance from the ground. With an increase in the flight duration, the probability of outage increases due to the reduction of transmission time.

\begin{figure}[t]
	\begin{center}
		\includegraphics[scale=.55]{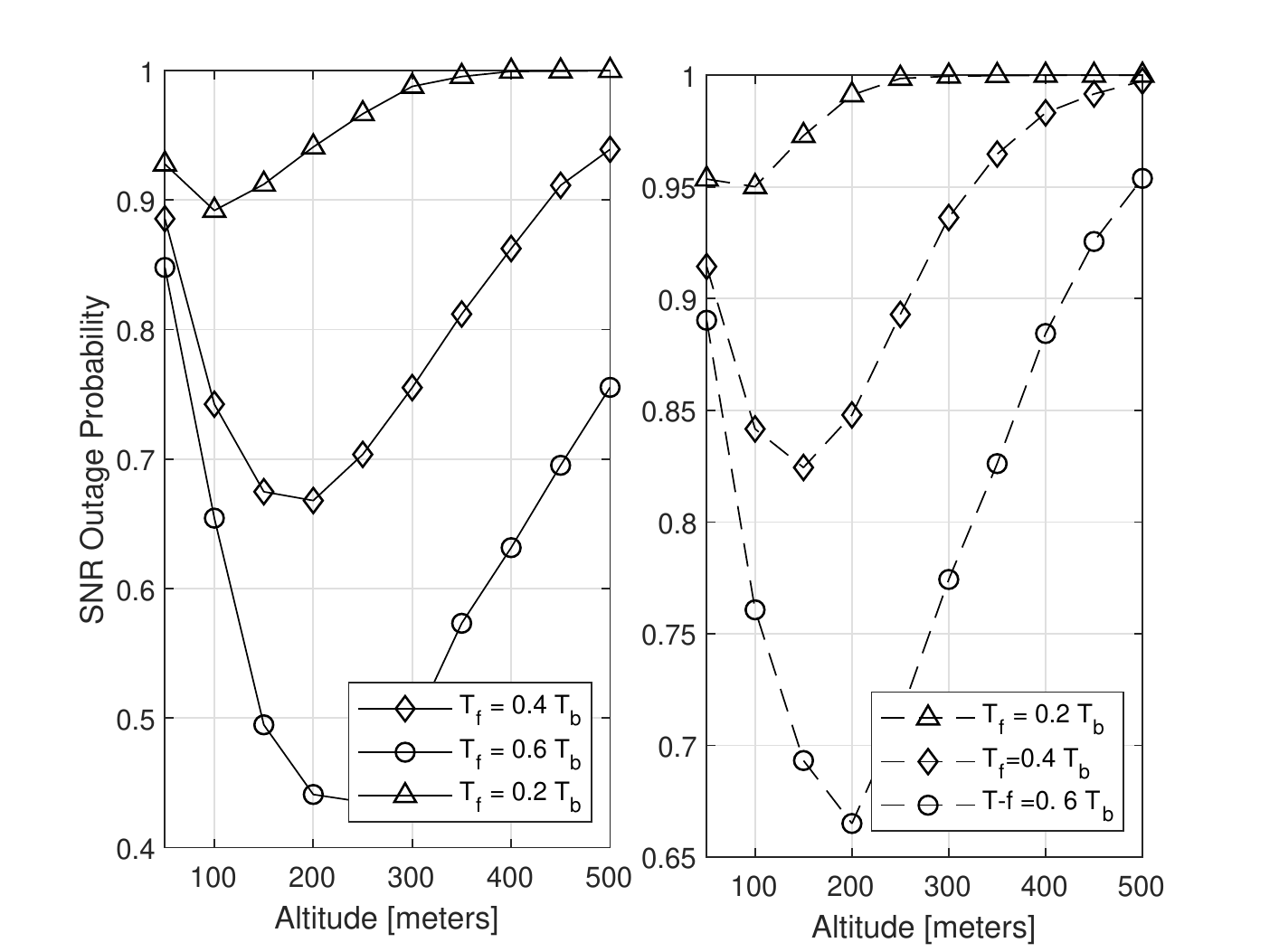}
	\end{center}
	\caption{Energy outage and rate outage as a function of the  altitude of the UAV $h$ and the flight duration $T_f$. }
	\label{fig7}
\end{figure}

\subsubsection{Impact of the  Flight Time}
\begin{figure}[t]
	\begin{center}
		\includegraphics[scale=.5]{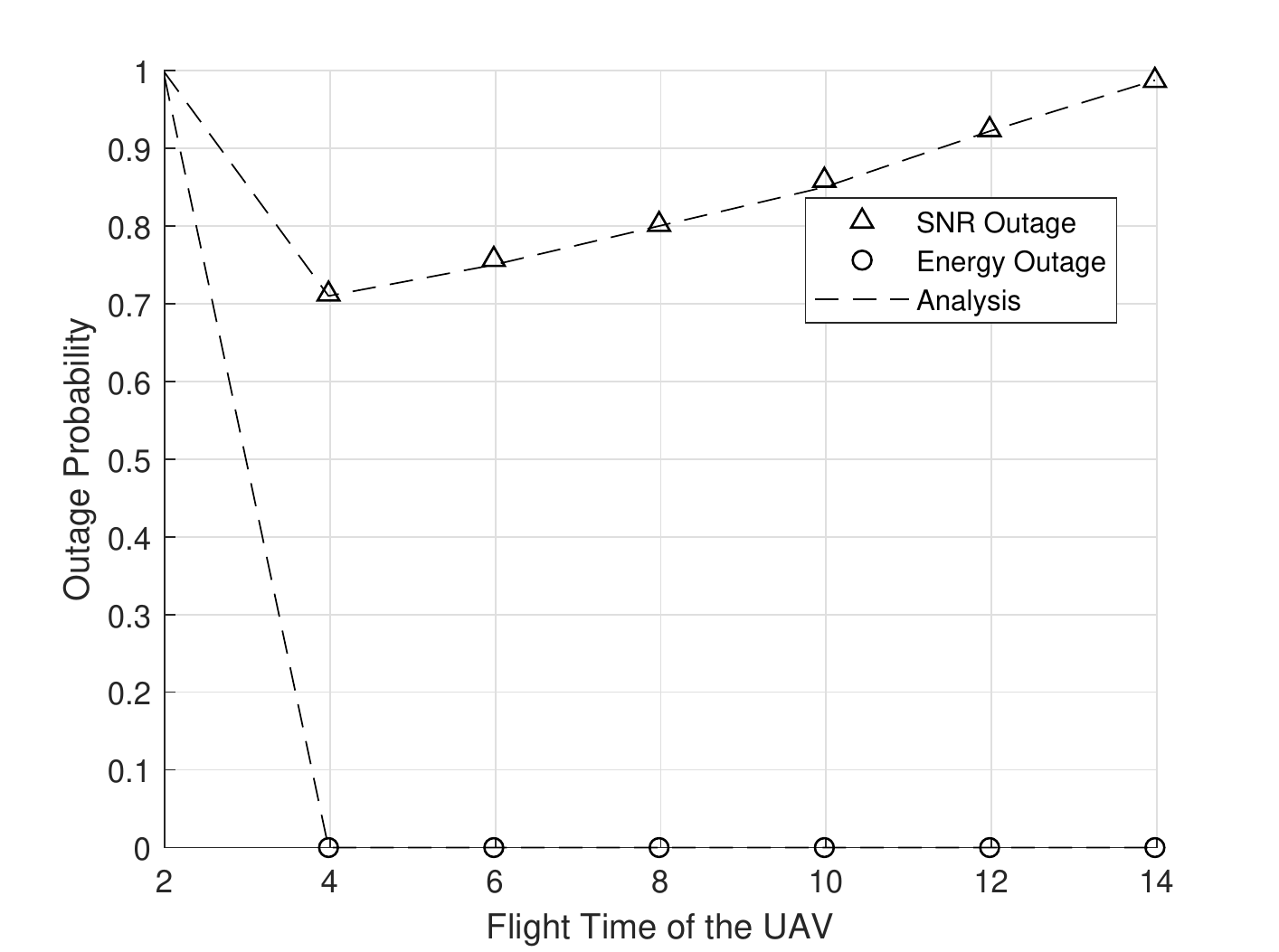}
	\end{center}
	\caption{Energy outage and rate outage as a function of the  transmission power of the UAV $P_d$ and the flight duration $T_f$. }
	\label{fig8}
\end{figure}

In \figref{fig8}, the energy outage and the SNR outage probabilities are plotted against the flight time of the UAV. We observe that with an increase in the flight time, the energy outage probability reduces since a longer flight time implies longer harvesting and shorter transmission durations which in turn reduces energy consumption.   On the other hand, with increasing flight time, the SNR outage first decreases and after a point it starts to increase.  That is, an optimal flight time exists. If the flight time is really small, then the harvesting time becomes limited and most of the energy consumption occurs during transmission, which results in energy outage, and subsequently,  in SNR outage. On the other hand, if the flight time is long, the transmission time becomes limited resulting in an SNR outage.

\subsubsection{Battery Recharge Time}
\begin{figure}[t]
	\begin{center}
		\includegraphics[scale=.5]{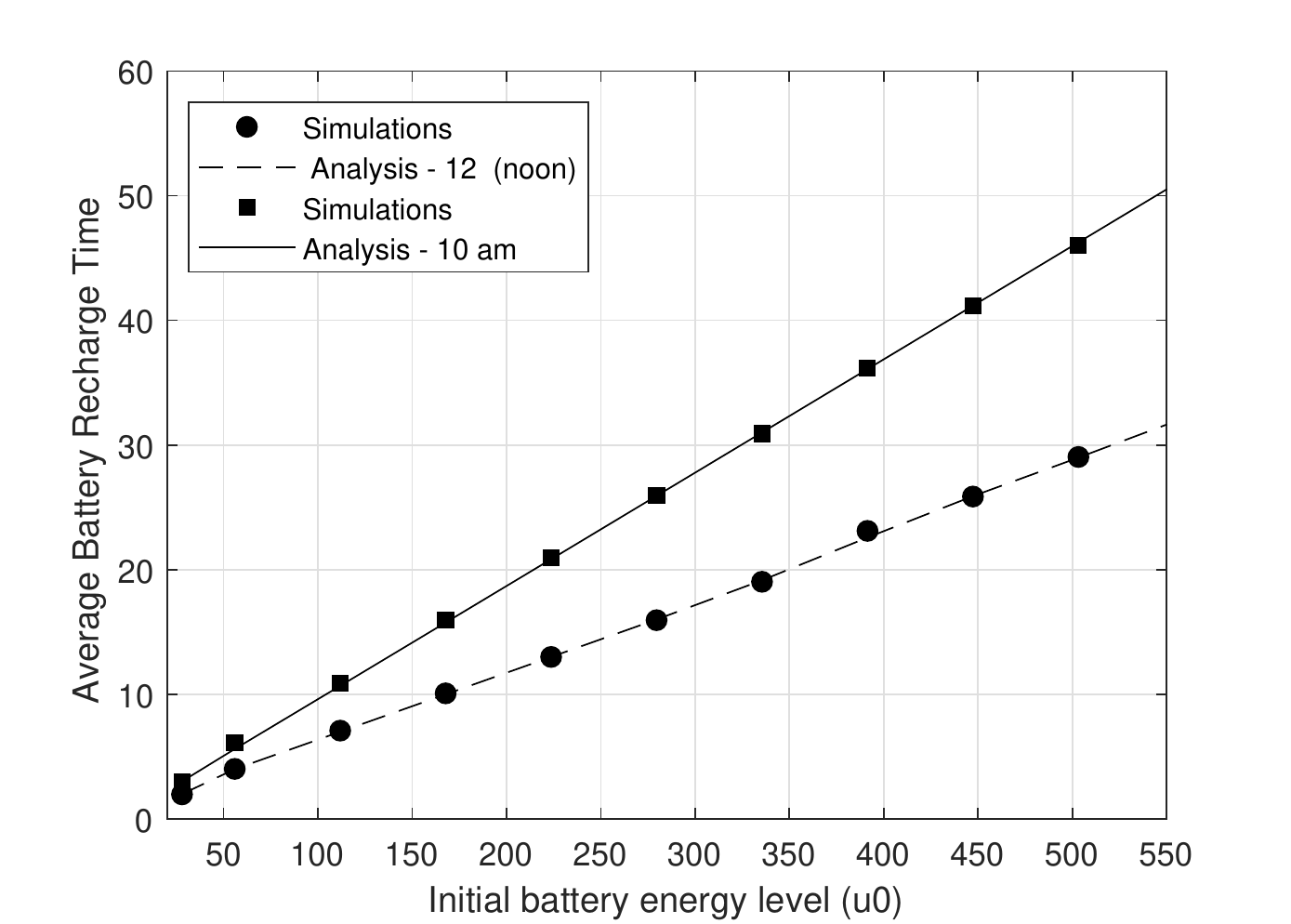}
	\end{center}
	\caption{Energy outage and rate outage as a function of the  transmission power of the UAV $P_d$ and the flight duration $T_f$. }
	\label{fig9}
\end{figure}

\figref{fig9} depicts the time to charge the battery of the UAV as a function of the level upto which the battery needs to be charged. Here, the energy packet size follows distribution of harvested power from solar energy with finite mean and variance. Here the inter arrival
time is exponentially distributed due to Poisson arrival consideration. The results from the simulations match closely with the theoretical prediction. We can observe that the more the level of energy, the more time it takes for the UAV battery to recharge which is intuitive. However, battery charging time is also affected by the harvesting time of the day. We can observe that, at noon, the harvested solar energy is comparatively higher, and therefore, the time to charge the battery  to the same level is reduced.

\begin{figure}[t]
	\begin{center}
		\includegraphics[scale=.45]{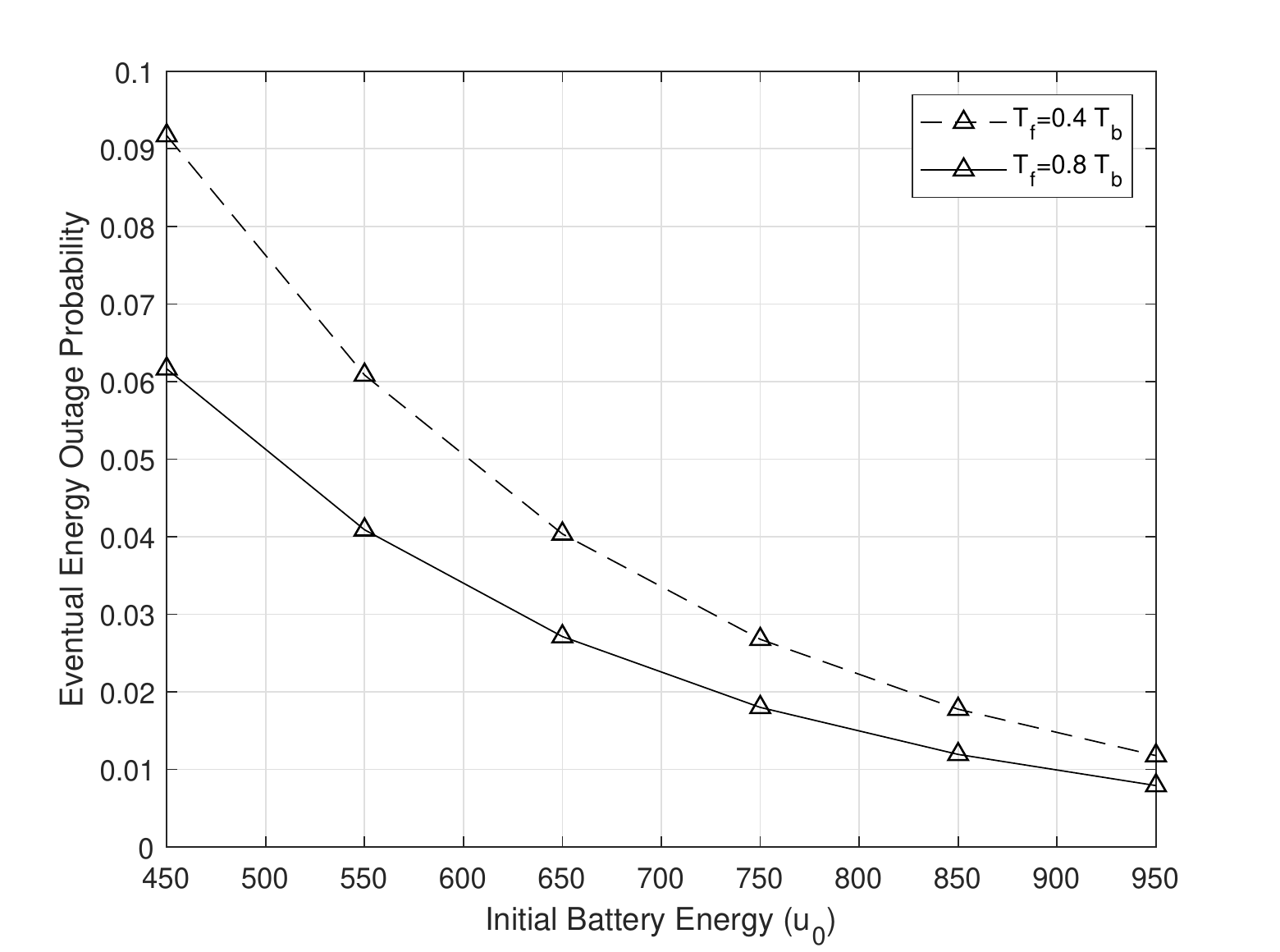}
	\end{center}
	\caption{Eventual energy outage  as a function of the  transmission power of the UAV and the initial UAV battery. }
	\label{fig10}
\end{figure}

In \figref{fig10}, the eventual energy outage is plotted against the initial battery energy of the UAV. It is evident  that, if we increase the initial battery energy, the eventual outage reduces due to the energy consumption required for supporting flight and transmission. We evaluate the eventual energy outage where  time is greater than or equal to $T_f$. Also, we can see that the eventual energy outage probability decreases significantly if $T_f$ increases due to more time devoted for energy harvesting and thus  accumulating more energy packets. Also, the energy consumption during flight is less than the energy consumption  during hover time as discussed earlier.

\begin{figure}[t]
	\begin{center}
		\includegraphics[scale=.65]{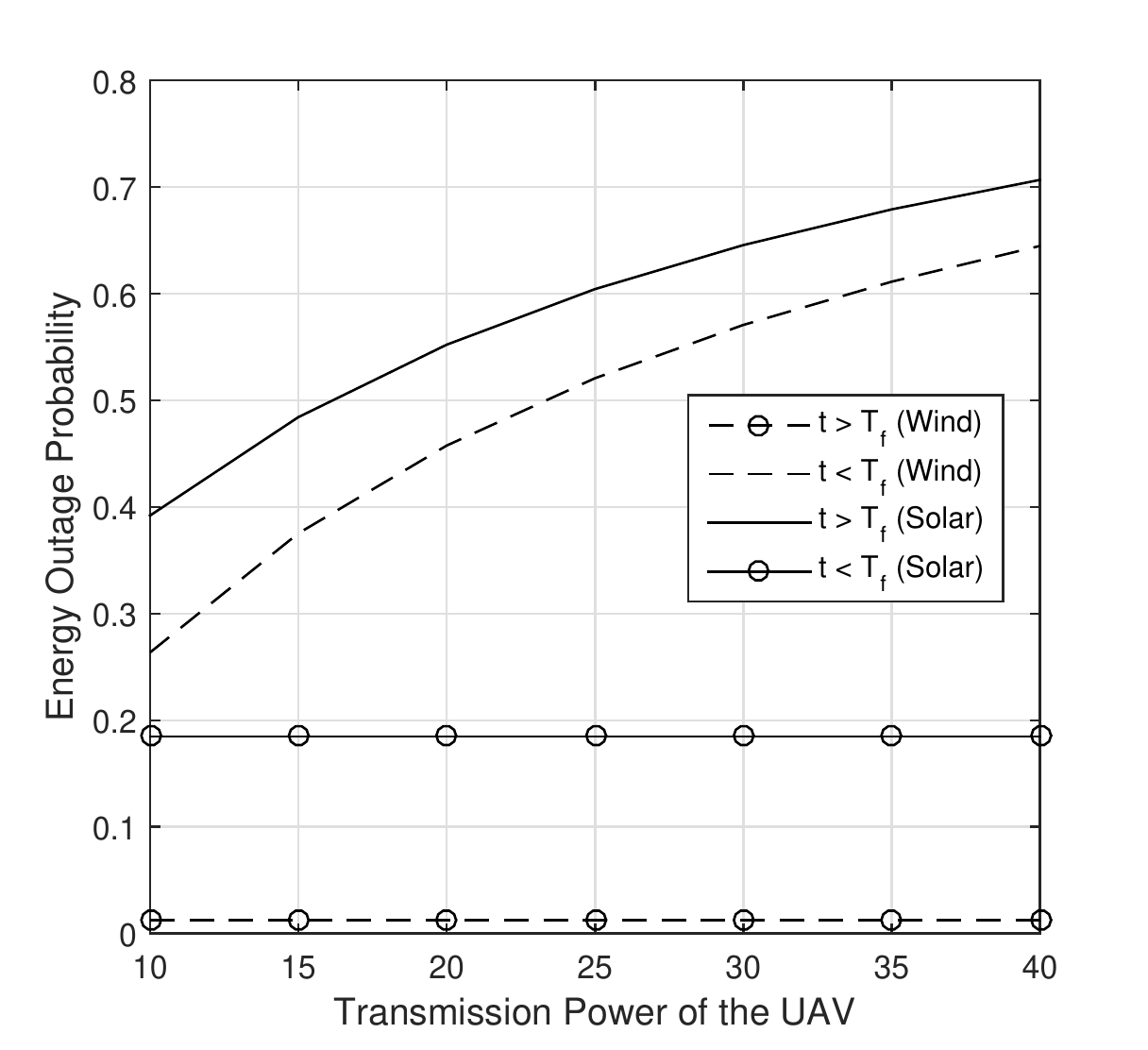}
	\end{center}
	\caption{Energy outage  as a function of the  transmission power of the UAV [$V_{\mathrm{ci}} =1, V_r=8$].}
	\label{fig11}
\end{figure}
In \figref{fig11}, the energy outage is plotted against the transmit power of the UAV. For $t>T_f$, if we increase the transmit power, energy consumption from the UAV battery increases which in result increases the energy outage.  We observe that both solar and wind energy outage follows similar pattern. For any time $t < T_f$, the energy outage is constant as there is no transmission during the flight. Also, the energy consumption during flight is less than the consumption  during transmission therefore we observe reduced energy outage probability.

Note that the energy outage in \figref{fig11} is different from the energy outage in \figref{fig6}. In \figref{fig11}, both the harvested power and energy arrival times are random whereas in \figref{fig6} only the harvested power is random. However, it can be noticed that the energy outage tends to be higher in \figref{fig11} as the harvested power (or equivalently energy packets) are arriving intermittently rather than continuously throughout the flight duration as is the case in \figref{fig6}. It can also be noted that, since solar energy is more deterministic compared to wind, the average energy outage in \figref{fig6} has a sharp transition. In \figref{fig11}, the cumulative solar energy during flight time is more stochastic due to not only the amount of harvested power but also the arrival times. As such, the energy outage increases gradually rather than sharply.

\section{Conclusion}
We  have derived the exact and novel  statistical models for the three renewable energy harvesting scenarios, i.e. harvested solar power, harvested wind power, and  hybrid solar and wind power, which have been verified by the simulations. Based on the developed models, we have derived the energy outage probability of the UAV and the SNR outage probability of the user. The impacts of the transmit power and flight time on the SNR and energy outage have been observed to obtain useful design guidelines for the optimal transmit power and flight duration of the UAVs. In addition, we have shown the novel applications of the moments of the harvested solar and wind power in characterizing performance measures such as the probability of charging UAV battery in the flight time, the average battery charging time, and the eventual energy outage probability. Numerical results demonstrate that an optimal transmit power or flight time can be determined to minimize the energy and SNR outage, respectively. Also, the results demonstrate the distinct stochastic behavior of solar and wind harvested power as the solar energy is relatively more deterministic compared to wind  energy. The developed statistical models for  harvested solar and wind power can as well be used for analyzing the performances of other aerial wireless communications systems as well as integrated terrestrial-aerial communications systems using energy harvesting.


\appendices
\renewcommand{\theequation}{A.\arabic{equation}}
\setcounter{equation}{0}
\section*{Appendix A: Proof of Theorem~1}
Given $I(t)$ is a shifted Gaussian random variable, the harvested solar power of the PV system $P$ at time $t$ in Section II can be rewritten as:
\begin{equation*}
\begin{split}
P= \frac{\eta_c}{K_c} I^2 \mathbb{U}\left(K_c-I\right) + \eta_c I \mathbb{U} \left(I-K_c\right),\\
\stackrel{(a)}{=}\frac{\eta_c}{K_c} I \left(I-K_c\right)\mathbb{U}(K_c-I) + \eta_c I,
\end{split}
\end{equation*}
where (a) utilizes the property $\mathbb{U} \left(K_c - I\right)=1-\mathbb{U} \left(I-K_c\right)$.
For notational simplicity, we skip the argument $t$, i.e. we consider $I(t)=I$. Note that $P$ is a function of $I$; therefore, we define $P=g(I)$. Since Dirac delta functions can be used to evaluate the PDF of the transformed random variables, we can derive $f_P(p)$ as follows:
\begin{align}
&f_P(p)=\int_0^{\infty} \delta \left(p-g(I)\right) f(I) dI,
\nonumber\\
&\Scale[1]{=\int_0^{\infty} \delta \left(p-\frac{\eta_c I \left(I-K_c\right)}{K_c} \mathbb{U}(K_c-I) - \eta_c I\right) f(I) dI}.
\label{A}
\end{align}
Note that the integral in \eqref{A} can be decomposed into two integrals as shown below:
\begin{equation}
 \begin{split}
f_P(p)=\int_0^{K_c} \delta \left[p-\frac{\eta_c}{K_c} I \left[I-K_c\right] + \eta_c I\right] f(I) dI
\\+ \int_{K_c}^{\infty} \delta(p-\eta_c I)f(I) dI.
\end{split}
\label{Bc}
\end{equation}
The first part of \eqref{Bc} can be simplified as follows:
\begin{align*}
&\Scale[1]{\int_0^{K_c} \delta \left[p-\frac{\eta_c}{K_c} I^2\right] f(I) dI,
=\int_0^{K_c} \delta \left[\frac{\eta_c}{K_c} \left(I^2-\frac{K_c}{\eta_c}p \right)\right] f(I) dI},
\\&\stackrel{(a)}{=}
\Scale[1]{\int_0^{K_c} \frac{K_c}{\eta_c} \delta \left[ \left(I^2-\frac{K_c}{\eta_c}p \right)\right] f(I) dI},
\\&\stackrel{(b)}{=}
\Scale[1]{\frac{K_c}{2 \eta_c} \int_0^{K_c} \sqrt{\frac{\eta_c}{p K_c}} \left(\ \delta(I+\sqrt{\frac{p K_c}{\eta_c}}) + \delta(I-\sqrt{\frac{p K_c}{\eta_c}}) \right) f(I) dI},
\\&\stackrel{(c)}{=}
\Scale[1]{\frac{K_c}{2 \eta_c} \int_0^{K_c} \sqrt{\frac{\eta_c}{p K_c}}\delta(I-\sqrt{\frac{p K_c}{\eta_c}})  f(I) dI},
\\&{=}
\Scale[1]{\frac{1}{2} \int_0^{\infty} \sqrt{\frac{K_c}{p \eta_c}}\delta(I-\sqrt{\frac{p K_c}{\eta_c}}) \mathbb{U}(K_c - I) f(I) dI},
\\&\stackrel{(d)}{=}
\Scale[1]{\frac{1}{2}  \sqrt{\frac{K_c}{p \eta_c}}\mathbb{U}(K_c - \sqrt{\frac{p K_c}{\eta_c}}) f(\sqrt{\frac{p K_c}{\eta_c}})},
\label{B}
\end{align*}
where (a) is obtained by applying the properties of Dirac delta functions $\delta(ax)=a^{-1}\delta(x)$ and (b) is obtained by applying the delta function property $\delta(x^2-a^2)=\frac{1}{2a}[\delta(x-a)+\delta(x+a)]$.
Note that radiation intensity $I$ cannot be negative, therefore, we discard the negative root and (c) can thus be obtained. Finally, the integral can be solved using the properties of Dirac-delta function as shown in (d).

The second part of \eqref{Bc} can be rewritten as follows: 
\begin{align}
\int_{0}^{\infty} \delta (p-\eta_c I) \mathbb{U}(I-K_c) f(I) dI,\\
\stackrel{(a)}{=}\frac{1}{\eta_c} f\left(\frac{p}{\eta_c}\right) \mathbb{U} (\frac{p}{\eta_c}-K_c),
\label{aa}
\end{align}
where (a) is obtained by solving the integral.
The distribution of the harvested solar power can thus be derived by combining the results in \eqref{B} and \eqref{aa}  as is shown in {\bf Theorem~1}.

\setcounter{equation}{0}
\renewcommand{\theequation}{B.\arabic{equation}}
\section*{Appendix B: Proof of Theorem~2}
Similar to the proof of {\bf Theorem~1}, the distribution of output power in the range $V_{ci}<v \leq V_r$ can be derived using the single random variable transformation,  as follows:
\begin{equation}
f_{P_W}(p_w)=\int_{V_{ci}}^{V_r} \delta \left(p_w-a \bar v^3\right)f_{\bar V}(\bar v),
\label{wind1}
\end{equation}
where $a=\frac{1}{2} \rho A C$. It is noteworthy that $f_V(v)$ has been truncated first over the range $V_{ci}<v \leq V_r$ yielding $f_{\bar V}(\bar v)$ and then is used in the aforementioned expression. That is,
\begin{equation}
f_{\bar V}(\bar v)=\frac{f_V(v)}{F_V(V_{ci})-F_V(V_{r})}.
\end{equation}
Now we apply the Dirac-delta function property in \eqref{wind1}   
\begin{equation*}
\delta(f(x))=\sum\limits_i \frac{\delta(x-a_i)}{|\frac{df}{dx}(a_i)|},
\end{equation*}
and rewrite \eqref{wind1} as follows:
\begin{equation}
\int_{V_{ci}}^{V_r} \frac{\delta(\bar v-v_1)}{\frac{3}{2} \rho A C \bar v^2} f_{\bar V}(\bar v) d {\bar v}, 
\end{equation}
where $v_1={(\frac{ p_w}{a})}^{1/3}$ is the only real root. After solving the integral, the PDF of output power harvested from wind energy can be given as follows:
\[f_P(p)= 
\begin{cases}
\frac{2 f_{\bar V}( v_1)}{3\rho A C {v_1}^2}, &  a V_{{ci}}^3 < p_w \leq a V^3_{{r}}\\
 \delta(p-P_r),   & a V^3_{{r}} <p_w \leq a V^3_{{co}}\\
 \delta(p) , & \mathrm{otherwise}\\
\end{cases},
\]
The PDF of the harvested power from wind energy can thus be given as in {\bf Theorem~2}.

\setcounter{equation}{0}
\renewcommand{\theequation}{C.\arabic{equation}}
\section*{Appendix C: Proof of Theorem~3: Part (a)}
Using {\bf Theorem~1}, the PDF of the solar power $f_P(p)$ can be given.
The Laplace transform $\mathcal{L}_P(s)= \int\limits_0^\infty e^{-sp}f_P(p) dp$ can thus be expressed as follows:
\begin{align}
\mathcal{L}_P(s) =\frac{\int\limits_{\eta_c K_c}^\infty e^{-sp-\frac{\left(\frac{p}{\eta_c}-I_d\right)^2}{2}}  }{\eta_c \sqrt{2 \pi}} +\frac{ \int\limits_0^{\eta_c K_c} \sqrt{\frac{1}{p}}e^{-sp -\frac{\left(\sqrt{\frac{K_c p}{\eta_c}}-I_d\right)^2}{2}}}{2 \sqrt{2 \pi \eta_c/K_c}}.
\label{Lap1}
\end{align}
Using \cite[Eq. 3.322/1]{ebook}, i.e.
$
\int\limits_u^\infty e^{-\frac{x^2}{4\beta}-\gamma x} dx= \sqrt{\pi \beta} e^{\beta \gamma^2} [1-\mathrm{erf}(\gamma \sqrt{\beta}+\frac{u}{2\sqrt{\beta}})]
$
and taking $\beta=\eta_c^2/2$ and $\gamma= s-\frac{I_d}{\eta_c}$, the first integral in \eqref{Lap1} can be solved in closed-form as follows:
\begin{equation}
\frac{1}{2} e^{-\frac{{I_d}^2}{2}+\frac{1}{2}(I_d-s \eta_c)^2} \mathrm{erfc}\left[\frac{-I_d+K_c+s\eta_c}{\sqrt{2}}\right]. 
\label{close1}
\end{equation}
The second integral in \eqref{Lap1} can be rewritten by changing variables
$\sqrt{p} =x $ as shown below:
\begin{equation}
\sqrt{\frac{K_c}{2 \pi \eta_c }} e^{-\frac{{I_d}^2}{2}}\int\limits_0^{\sqrt{\eta_c K_c}}e^{-(\frac{K_c}{2 \eta_c}+s)x^2+\sqrt{\frac{K_c}{\eta_c }}x I_d} dx.
\label{lap2}
\end{equation}
We note that $\int\limits_0^u f(x)dx = \int\limits_0^\infty f(x)dx - \int\limits_u^\infty f(x)dx$; therefore, we split \eqref{lap2} into two integrals. Then we use  $
\int\limits_0^\infty e^{-\frac{x^2}{4\beta}-\gamma x} dx= \sqrt{\pi \beta} e^{\beta \gamma^2} [1-\mathrm{erf}(\gamma \sqrt{\beta})]
$~\cite[Eq. 3.322/2]{ebook} as well as \cite[Eq. 3.322/1]{ebook} to solve the integrals by taking $\beta=\frac{\eta_c}{2(K_c+2\eta_c s)}$ and $\gamma= -\sqrt{\frac{K_c}{\eta_c}}I_d$. The closed-form for \eqref{lap2} can thus be expressed as follows: 
\begin{equation}
\frac{e^{-\frac{{I_d}^2}{2} + \frac{ I_d^2}{ 2+\frac{4 s \eta_c}{K_c}}}\left(
\mathrm{erf}\left[ \frac{I_d}{\sqrt{2+\frac{4 s \eta_c}{K_c}}}\right]
+ \mathrm{erf}\left[ \frac{I_d-K_c-2 s \eta_c}{\sqrt{2+\frac{4 s \eta_c}{K_c}}}\right]\right)}{2 \sqrt{2s+K_c/\eta_c}}. 
\label{close2}
\end{equation}
Combining \eqref{close1} and \eqref{close2}, the MGF can be expressed as shown in {\bf Theorem~3}.

\setcounter{equation}{0}
\renewcommand{\theequation}{D.\arabic{equation}}
\section*{Appendix~D: Proof of Theorem~3: Part (b)}
The Laplace transform of the harvested wind power  can be derived using $f_{P_w}(p_w)$ given in  {\bf Theorem~2} as follows:
\begin{align*}
\mathcal{L}_{P_w}(s)&=\frac{k a^{-\frac{k}{3}}}{3 c^k} \int\limits_0^\infty e^{- s p} {p_w}^{\frac{k-3}{3}} \mathrm{exp}\left(-\frac{{p_w}^{\frac{k}{3}}}{a^{\frac{k}{3}}c^k}\right)dp,
\nonumber\\
&\stackrel{(a)}{=}\frac{k}{3} \int\limits_0^\infty e^{-s a c^3 x} e^{-x^{\frac{k}{3}}} x^{\frac{k}{3}-1} dx,
\nonumber\\
&\stackrel{(b)}{=} \int\limits_0^\infty e^{-s a c^3 y^{\frac{3}{k}}} e^{-y} dy,
\nonumber\\
&\stackrel{(c)}{=} 
\int\limits_0^\infty \sum\limits_{n=0}^\infty \frac{(-s a c^3 y^{\frac{3}{k}})^n}{n!} e^{-y} dy,
\nonumber\\
&\stackrel{(d)}{=} 
\sum\limits_{n=0}^\infty \frac{(-s a c^3 )^n}{n!} \left(\frac{3n}{k}\right)!,
\label{mgf1}
\end{align*}
where (a) follows by taking $x=\frac{p}{a c^3}$, (b) follows by 
substituting $y=x^{k/3}$, (c) follows by using the identity $e^x=\sum\limits_0^\infty \frac{x^n}{n!}$, and (d) follows by  interchanging the summation and integration and solving the integral using $\int_0^\infty x^m e^{-\mu x} dx=m! \mu^{-m-1}$ from \cite[3.351/2]{ebook}. 

\bibliography{IEEEfull,References}
\bibliographystyle{IEEEtran}

\end{document}